\documentclass{article}

 \PassOptionsToPackage{numbers}{natbib}
    \usepackage[final]{neurips_2020}

\usepackage[utf8]{inputenc} 
\usepackage[T1]{fontenc}    
\usepackage{hyperref}       
\usepackage{url}            
\usepackage{booktabs}       
\usepackage{amsfonts}       
\usepackage{nicefrac}       
\usepackage{microtype}      

\title{Hedging in games: Faster convergence of\\ external and swap regrets}

\author{%
		 Xi Chen\thanks{Supported by NSF IIS-1838154 and NSF CCF-1703925.} \\
		 Columbia University\\
		 \texttt{xichen@cs.columbia.edu} \\
		 \And
	Binghui Peng\\
	Department of Computer Science\\
	Columbia University\\
	\texttt{bp2601@columbia.edu} \\
}

\usepackage{amsmath}
\usepackage{amsthm}
\usepackage{amssymb}
\usepackage{algorithm}
\usepackage{subfig}
\usepackage{color}
\usepackage[english]{babel}
\usepackage{graphicx}
\usepackage{grffile}

\usepackage{wrapfig,epsfig}
\usepackage{epstopdf}
\usepackage{url}
\usepackage{color}
\usepackage{epstopdf}
\usepackage{algpseudocode}
\usepackage{multirow}
\usepackage[T1]{fontenc}
\usepackage{bbm}
\usepackage{comment}
\usepackage{dsfont}

\usepackage{tikz}
\hypersetup{colorlinks=false,citecolor=red,linkcolor=blue}
\usetikzlibrary{arrows}

\newtheorem{theorem}{Theorem}[section]
\newtheorem{lemma}[theorem]{Lemma}
\newtheorem{definition}[theorem]{Definition}

\newtheorem{corollary}[theorem]{Corollary}

\newtheorem{remark}[theorem]{Remark}
\newtheorem{claim}[theorem]{Claim}

\newcommand{\wt}{\widetilde}

\newcommand{\eps}{\epsilon}

\newcommand{\R}{\mathbb{R}}

\renewcommand{\varepsilon}{\epsilon}
\renewcommand{\tilde}{\wt}

\renewcommand{\eps}{\epsilon}

\newcommand{\mT}{\mathcal{T}}

\DeclareMathOperator*{\E}{{\mathbb{E}}}

\DeclareMathOperator{\OPT}{OPT}

\newcommand{\bx}{\mathbf{x}}

\newcommand{\bs}{\mathbf{s}}

\DeclareMathOperator{\regret}{regret}
\DeclareMathOperator{\sregret}{swap-regret}

\def\BMHedge{\texttt{BM-Optimistic-Hedge}}

\definecolor{b2}{RGB}{51,153,255}
\definecolor{mygreen}{RGB}{80,180,0}
\definecolor{yanglav}{rgb}{0.67, 0.38, 0.8}

\renewcommand{\Xi}[1]{{\color{red}[Xi: #1]}}

\makeatletter
\newcommand*{\RN}[1]{\expandafter\@slowromancap\romannumeral #1@}
\makeatother

\def\DKL{D_{\text{KL}}}

\begin{document}


\maketitle
\begin{abstract}

We consider the setting where players run the Hedge algorithm or
its optimistic variant to play an $n$-action game repeatedly for $T$ rounds.\vspace{0.03cm}
\begin{flushleft}\begin{itemize}
		\item 
		For two-player games, we show that the regret of
		optimistic Hedge decays at rate $\smash{ {O}( 1/T ^{5/6} )}$, improving the previous bound 
		of $\smash{O(1/T^{3/4})}$ by Syrgkanis, Agarwal, Luo and  Schapire \cite{syrgkanis2015fast}.\vspace{0.07cm}


		\item 
		In contrast, we show that the convergence rate of vanilla Hedge is no better than 
		$\smash{ O(1/ \sqrt{T})}$, addressing an open question posed in Syrgkanis, Agarwal, Luo and  Schapire~\cite{syrgkanis2015fast}.\vspace{0.03cm}
	\end{itemize}\end{flushleft}
	For general $m$-player games, we show that the swap regret of each player decays at 
	$\smash{{O}(m^{1/2} (n\log n/T)^{3/4})}$ when they combine optimistic Hedge
	with the classical external-to-internal reduction of Blum and Mansour \cite{blum2007external}.
	Via standard connections, our new (swap) regret bounds imply faster convergence to coarse correlated equilibria
	in two-player games and to correlated equilibria in multiplayer games.



\end{abstract}


\section{Introduction}
\label{sec:intro}

Online algorithms for regret minimization 
  play an important role 
  in many applications in machine learning 
  where real-time sequential decision making is crucial \cite{hazan2016introduction, cesa2006prediction,shalev2011online}.
A number of~algorithms have been developed, including 
  Hedge\hspace{0.05cm}/\hspace{0.05cm}Multiplicative Weights \cite{arora2012multiplicative}, Mirror Decent \cite{hazan2016introduction}, 
  Follow the Regularized\hspace{0.05cm}/\hspace{0.05cm}Perturbed Leader \cite{kalai2005efficient}, and their power and limits against an adversarial
  environment have been well understood:
The average (external) regret decays at a rate of $\smash{O(1/\sqrt{T})}$ after $T$ rounds, 
  which is known to be tight for any online algorithm.

What happens if players in a repeated game run one of these algorithms?
Given that 
  they~are~now running against similar algorithms over a fixed game,
  could the regret of each player  decay~signifi\-cantly faster than $\smash{ 1/\sqrt{T}  }$? 
This was answered positively in a sequence of works \cite{daskalakis2011near,rakhlin2013optimization, syrgkanis2015fast}.
Among these results, the one that is most relevant to ours is that of 
  Syrgkanis, Agarwal, Luo and Schapire \cite{syrgkanis2015fast}. They showed that 
  if every player in a multiplayer game
  runs an algorithm that satisfies the  RVU (Regret bounded by Variation 
  in Utilities) property, then the regret of each player decays at $\smash{O(1/T^{3/4})}$. 
\emph{Can this bound be further improved?}

Besides regret minimization,
  understanding no-regret dynamics in games is motivated by 
    connections with various equilibrium concepts
     \cite{freund1996game,foster1997calibrated,foster1999regret,hart2000simple, blum2007external, greenwald2008more,nisan2007algorithmic}.
For example, if every player runs an algorithm with 
  vanishing regret, then the empirical distribution must converge to a \emph{coarse} 
  correlated equilibrium~\cite{cesa2006prediction}.
Nevertheless, to converge to a more preferred correlated equilibrium \cite{Aumann74},  
  a stronger notion of regrets called \emph{swap regrets} (see Section \ref{sec:preliminary}) is required \cite{foster1997calibrated,hart2000simple,blum2007external}.
The minimization of swap regrets under the adversarial setting
  was studied by Blum and Mansour \cite{blum2007external}.
They gave a generic reduction from regret minimization algorithms which led to
  a tight $\smash{O(\sqrt{n\log n/T})}$-bound for the average swap regret.
A natural question is \emph{whether a speedup similar to that of} \cite{syrgkanis2015fast}
  \emph{is possible for swap regrets in the repeated game setting.}

\textbf{Our contributions: Faster convergence of swap regrets.}
We give the first algorithm 
  that achieves an average swap regret that is significantly 
  lower than $\smash{O(1/\sqrt{T})}$ under the repeated game setting.
This algorithm, denoted by \BMHedge, combines the external-to-internal reduction 
  of \cite{blum2007external} with the optimistic Hedge algorithm \cite{rakhlin2013optimization, syrgkanis2015fast} 
   as its regret minimization component. 
(Optimistic Hedge can be viewed as an instantiation of the optimistic Follow the Regularized
  Leader algorithm; see Section \ref{sec:preliminary}.) 
We show that if every player in a repeated game of $m$ players and $n$ actions 
  runs \BMHedge, then the average swap regret is at most
  $\smash{O(m^{1/2} (n\log n/T)^{3/4})}$; see Theorem 
  \ref{thm:correlated-equilibrium} in Section \ref{sec:correlated}.
Via the relationship between correlated equilibria and swap regrets,
  our result implies faster convergence to 
  a correlated equilibrium. 
When specialized to two-player games, the empirical distribution of players running $\BMHedge$
  converges to an $\eps$--correlated equilibrium after
  $\smash{O(n\log n/\eps^{4/3})}$ rounds,
  improving the $O(n\log n/\eps^2)$ bound of \cite{blum2007external}.

Our main technical lemma behind Theorem \ref{thm:correlated-equilibrium} shows that strategies produced 
  by the algorithm of~\cite{blum2007external} with optimistic Hedge moves very slowly in $\ell_1$-norm under the 
  adversarial setting (which in turn allows us to apply a stability argument similar to  
  \cite{syrgkanis2015fast}).
This came as a surprise because a key component of the algorithm of \cite{blum2007external} each round is to 
  compute the stationary distribution of a Markov chain, which is highly sensitive 
  to small changes in the Markov chain.
We overcome this difficulty by exploiting the fact that Hedge only 
  incurs small \emph{multiplicative} changes to the Markov chain, which allows us to bound
  the change in the stationary distribution using the classical Markov chain tree theorem.
{We further demonstrate the power of this technical ingredient by deriving another fast no-swap regret algorithm, based on a folklore algorithm in \cite{cesa2006prediction} and optimistic predictions
  (see Appendix \ref{sec:extension-phi}).
Both of these two algorithms enjoy the benefits of faster convergence when playing with each other, while remain robust against adversaries (see Corollary \ref{cor:robus-adv} in Appendix \ref{appendixsec}).}

\textbf{Our contributions: Hedge in two-player games.}
In addition we consider regret minimization in a two-player game with $n$ actions
  using either vanilla or optimistic Hedge.
We show that optimistic Hedge can achieve an average regret of $\smash{ {O}(1/T^{5/6})}$, improving the bound   $\smash{O(1/T^{3/4})}$ by \cite{syrgkanis2015fast} for two-player games; 
  see Theorem \ref{thm:coarse-correlated} in Section \ref{sec:coarse}.
In contrast, we show that even under this game-theoretic setting,
  vanilla Hedge cannot asymptotically outperform the $\smash{O(1/\sqrt{T})}$ adversarial bound;
  see Theorem \ref{thm:lower} in Section \ref{sec:lower-bound}.
This addresses an open question posed by \cite{syrgkanis2015fast} concerning the convergence
  rate of vanilla Hedge in a repeated game.

The key step in our analysis of optimistic Hedge is to show that, even under the    adversarial setting,
  the trajectory length of strategy movements (in their squared $\ell_1$-norm) 
  can be bounded using that 
  of cost vectors (in $\ell_\infty$-norm); see Lemma \ref{lem:movement}.
(Intuitively, it is unlikely for the strategy of optimistic Hedge to 
  change significantly over time while the loss vector stays stable.)
This allows us to build a strong relationship 
  between the trajectory length of each player's strategy movements,
  and then use the RVU property of optimistic Hedge to bound their individual regrets. 

Our lower bounds for vanilla Hedge use three very simple $2\times 2$ games
  to handle different ranges of the learning rate $\eta$.
For the most intriguing case when $\eta$ is at least $\smash{\Omega(1/\sqrt{n})}$ 
and bounded 
  from above by some constant, we study the zero-sum Matching Pennies game and 
  use it to show that the overall regret of at least 
  one player is $\smash{\Omega(\sqrt{T})}$.
Our analysis is inspired by the result of \cite{bailey2018multiplicative}  which shows that the KL divergence
  of strategies played by Hedge in a two-player zero-sum game is strictly increasing.
For Matching Pennies, we start with a quantitative bound on how fast the KL divergence
  grows in Lemma \ref{lem:kl-lower}.
This implies the existence of a window of length $\smash{\sqrt{T}}$ during which
  the cost of one of the player grows by $\smash{\Omega(1)}$ each round;
  the zero-sum structure of the game allows us to conclude that at least
  one of the players must have regret at least $\smash{\Omega(\sqrt{T})}$ 
  at some point in this window. 





\subsection{Related work}



Initiated by Daskalakis, Deckelbaum and Kim~\cite{daskalakis2011near}, there has been a sequence of works that study no-regret learning algorithms in games~\cite{rakhlin2013optimization, syrgkanis2015fast, foster2016learning, wei2018more}. 
Daskalakis~et.~al. \cite{daskalakis2011near} designed an algorithm by adapting Nesterov's 
accelerated saddle point algorithm to two-player {\em zero-sum games}, and showed that if 
  both players run this algorithm then their average regrets decay at rate ${O}(1/T)$, which is optimal.
Later Rakhlin and Sridharan~\cite{rakhlin2013online, rakhlin2013optimization} developed a simple and intuitive family of algorithms, i.e. {\em optimistic Mirror Descent} and {\em optimistic Follow the Regularized Leader}, that incorporate predictions into the strategy. They proved that if both players adopt the algorithm, then their average regrets also decay at rate $O(1/T)$ in {\em zero sum games}.
Syrgkanis~et.~al.~\cite{syrgkanis2015fast} further strengthened this line of works by showing that in a {\em general} 
  $m$-player game, if every player runs an algorithm that satisfies the RVU property 
  then the average regret decays~at  rate
  $\smash{O(1/T^{3/4})}$. 
Syrgkanis et. al.~\cite{syrgkanis2015fast} also considered the convergence of social welfare and proved an 
  even faster rate of $O(1/T)$ in smooth games~\cite{roughgarden2015intrinsic}.
Foster et. al.~\cite{foster2016learning}  extended \cite{syrgkanis2015fast}  and showed that if one only aims for an approximately optimal social welfare, then the class of algorithms allowed can be much broader.
Recently, Daskalakis and Panageas~\cite{daskalakis2018last} proved the last iteration convergence of optimistic Hedge in zero-sum game, i.e., instead of averaging over the trajectory, they showed that optimistic Hedge converges to a Nash equilibrium in a zero-sum game. 
\nocite{daskalakis2018training, wang2018acceleration, rakhlin2013optimization}

There is  also a growing body of works~\cite{mertikopoulos2018cycles, bailey2018multiplicative, bailey2019fast, cheung2019vortices} on the dynamics of no-regret learning over games in the last few years.
Most of these works studied the dynamics of no-regret learning from a dynamical system point of view and provided qualitative intuition on the evolution of no-regret learning. Among them, \cite{bailey2019fast} is most relevant, in which Bailey and Piliouras proved an $\smash{\Omega(\sqrt{T})}$ lower bound on the convergence rate of online gradient descent~\cite{zinkevich2003online} for the $2 \times 2$ Matching Pennies game. 
However, we remark that their lower bound only works for online gradient descent 
  and they need to fix the learning rate $\eta$ to 1. Our lower bound for vanilla Hedge in
    two-player games holds for arbitrary learning rates.


\section{Preliminary}
\label{sec:preliminary}

\textbf{Notation.} Given two positive integers $n\le m$, 
we use $[n]$ to denote $\{1,\ldots,n\}$ and 
$[n:m]$ to denote $\{n,\ldots,m\}$. 
We use $\DKL(p\|q)$ to denote the KL divergence with natural logarithm.

\textbf{Repeated games and regrets.}
Consider a game $G$ played between $m$ players, where each
  player $i\in [m]$ has a strategy space $S_{i}$ with $|S_i| = n$ 
  and a \emph{loss} function $\mathcal{L}_i: S_1\times\cdots\times S_{m}\rightarrow [0, 1]$ 
  such that $\mathcal{L}_i(\bs)$ is the  {loss} of player $i$
  for each pure strategy profile $\bs = (s_1,\ldots, s_n)\in S_1\times \cdots \times S_m$.
A mixed strategy for player $i$ is a probability distribution $x_i$ over $S_i$, where the $j${th} action is played with probability $x_{i}(j)$.
Given a mixed (or pure) strategy profile $\bx=(x_1,\ldots,x_m)$ (or $\bs=(s_1,\ldots,s_m)$), we write 
  $\bx_{-i}$ (or $\bs_{-i}$) to denote the profile after removing $x_i$ (or $s_i$, respectively).

We consider the scenario where the $m$ players play $G$ repeatedly for
  $T$ rounds. 
At the beginning of each round $t$, $t\in [T]$,
  each player $i$ picks a mixed strategy $x_i^t$ and let 
  $\bx^t=(x_1^t,\ldots,x_m^t)$ be the mixed strategy profile. 
We consider the {\em full information} setting where each player
  observes the {\em expected} loss of {\em all} her actions. 
Formally, player $i$ observes a loss vector $\ell_i^t$ with 
  $\smash{\ell_{i }^{t}(j) = \E_{\bs_{-i} \sim \bx_{-i}^{t}}[ \mathcal{L}_i (j, \bs_{-i})]}$,
  and her expected loss is given by $\smash{\langle x_i^t, \ell_i^t \rangle}$.
At the end of round $T$, 
  the \emph{regret} of player $i$ is
\begin{align}
\label{eq:regret1}
\regret_{T}^{i} = \sum_{t\in [T]} \langle x_{i}^{t}, \ell_{i}^{t} \rangle - \min_{j \in [n]} \sum_{t\in [T]}\ell_{i }^{t}(j), 
\end{align}
i.e., the maximum gain one could have obtained by switching to some fixed action. 
A stronger notion of regret, referred as {\em swap regret}, is defined as
\begin{align}
\label{eq:regret2}
\sregret_{T}^{i} = \sum_{t\in [T]} \langle x_{i}^{t}, \ell_{i}^{t} \rangle - \min_{\phi} \sum_{t\in [T]}\sum_{j \in [n]}  x_{i }^{t}(j)\cdot \ell^t_{i}(\phi(j)), 
\end{align}
where the minimum is over all $n^n$ (swap) functions 
  $\phi: [n] \rightarrow [n]$ that swap action $j$ with $\phi(j)$.
The swap regret equals the maximum gain one could have achieved by using a fixed swap function
  over its past mixed strategies.
\textbf{Hedge.}
Consider the adversarial online model where a player has $n$ actions and picks a 
  distribution $x^t$ over them at the beginning of each round $t$.
During round $t$ the player receives a loss vector $\ell^t$ and pays a loss of $\langle x^t,\ell^t\rangle$.
The vanilla Hedge algorithm \cite{Hedge97} with learning rate $\eta>0$ starts by setting $x^1$ to be the uniform 
  distribution and then keeps applying the following updating rule to obtain $x^{t+1}$ from $x^t$ and 
  the loss vector $\ell^t$ at the end of round $t$: for each action $j\in [n]$,
$$
x^{t+1} (j) = \frac{x^{t} (j)\cdot \exp(-\eta\cdot \ell^{t} (j)) }{\sum_{k\in [n]} x^{t} (k)\cdot \exp(-\eta\cdot \ell^{t}(k) )}.
$$
On the other hand, the optimistic Hedge algorithm can be obtained from 
  the {\em optimistic follow the regularized leader} proposed by~\cite{rakhlin2013optimization, syrgkanis2015fast}, and have the following updating rule: 
\begin{align}
\label{eq:opt-hedge-update}
x^{t+1} (j) = \frac{x^{t} (j)\cdot \exp(-\eta (2\ell^{t} (j) - \ell^{t-1} (j))}{\sum_{k\in [n]} x^{t} (k)\cdot \exp(-\eta (2\ell ^{t}(k) - \ell ^{t-1}(k))},
\end{align}
with $\ell^0=\mathbf{0}$ being the all-zero vector.
We have the following regret bound for optimistic Hedge.
\begin{lemma}[\cite{rakhlin2013optimization, syrgkanis2015fast}]
	\label{lem:optimisticRegret}
Under the adversarial setting, optimistic Hedge satisfies
	\begin{align}
	\label{eq:optimistic}
	\regret_{T} \leq \frac{2\log n}{\eta} + \eta \sum_{t\in [T]}\|\ell^t - \ell^{t-1}\|_{\infty}^{2} - \frac{1}{4\eta}\sum_{t\in [T]}\|x^{t+1} - x^{t}\|_{1}^{2}.
	\end{align}
\end{lemma}




\section{Optimistic Hedge in Two-Player Games}
\label{sec:coarse}

In this section we analyze the performance of the optimistic Hedge algorithm
  when it is used by two players to play a (general, not necessarily zero-sum) $n\times n$ game repeatedly. 


\begin{theorem}
	\label{thm:coarse-correlated}
	Suppose both players in a two-player game run optimistic Hedge  for $T$ rounds with 
	learning rate $ \eta= ( \log n/{T} )^{1/6} $. Then the individual regret of each player is  
	$\smash{O(T^{1/6} \log^{5/6} n)}$.
\end{theorem}

We assume without loss of generality that $T\ge \log n$; otherwise,
  the regret of each player is trivially at most $\smash{T\le T^{1/6} \log^{5/6} n}$.
The following lemma is essential to our proof of Theorem \ref{thm:coarse-correlated}.
Consider the adversarial online setting where a player
  runs optimistic Hedge for $T$ rounds. 
The lemma bounds the trajectory length of the strategy movement using 
  that of cost vectors.
\begin{lemma}
	\label{lem:movement}
Suppose that a player runs optimistic Hedge with learning rate $\eta$
  for $T$ rounds.
Let $\ell^0,\ell^1,\ldots,\ell^T$ be the cost vectors with $\ell^0=\mathbf{0}$ and $x^1,\ldots,x^T$ be the strategies played.
Then
	\begin{align}
	\label{eq:coarse}
	\sum_{t\in [2:T]}\|x^t - x^{t - 1}\|_{1}^{2} \leq O(\log n) + O(\eta+\eta^2)
	  \sum_{t\in [T-1]}\|\ell^t - \ell^{t-1}\|_{\infty}.
	\end{align}
\end{lemma}

We delay the proof of Lemma \ref{lem:movement} to Appendix \ref{sec:movement} and
  use it to prove Theorem \ref{thm:coarse-correlated}.
\begin{proof}[Proof of Theorem~\ref{thm:coarse-correlated} assuming Lemma \ref{lem:movement}]
Let $G=(A,B)$ be the game, where $A,B\in [0,1]^{n\times n}$ denote the cost matrices of the first and 
  second players, respectively. 
We use $x^t$ and $y^t$ to denote strategies played by the two players 
  and use $\ell_x^t$ and $\ell_y^t$ to denote their cost vectors in the $t$th round.
So we have $\smash{\ell_x^t=Ay^t}$ and $\smash{\ell_y^t=B^{T}x^t}$.
Therefore, we have for each $t\ge 2$:
\begin{align}
	&~\|\ell^t_{y} - \ell^{t-1}_{y}\|_{\infty} = \|B^T(x^t - x^{t-1})\|_{\infty} \leq \|x^t - x^{t-1}\|_1
	\label{hehe100}\quad\text{and}\\[0.4ex]
	&~\|\ell^t_{x} - \ell^{t-1}_{x}\|_{\infty} = \|A(y^t - y^{t-1})\|_{\infty} \leq \|y^t - y^{t-1}\|_1.\nonumber
	\end{align}
Without loss of generality it suffices to bound the regret of the second player. 
Set $\eta=(\log n/T)^{1/6}$ with $T\ge \log n$ so that $\eta\le 1$. We have\vspace{0.1cm}
	\begin{align*}
	\regret_{T}^{y} \leq &~ \frac{2\log n}{\eta} + \eta \sum_{t\in [T]}\|\ell^t_{y} - \ell^{t-1}_{y}\|_{\infty}^{2} - \frac{1}{4\eta}\sum_{t\in [T]}\|y^{t+1} - y^{t}\|_{1}^{2} && \text{Lemma \ref{lem:optimisticRegret}}\\
	\leq &~ \frac{2\log n}{\eta} +
	\eta+ \eta \sum_{t\in [2:T]}\|x^t - x^{t-1}\|_{1}^{2} - 	\frac{1}{4\eta}\sum_{t\in [2:T+1]}\|\ell^{t}_{x} - \ell^{t-1}_{x}\|_{\infty}^{2} &&\text{using (\ref{hehe100})} \\
	\leq &~\frac{2\log n}{\eta} +\eta+ \eta \left(O(\log n) + O(\eta) \sum_{t\in [T-1]}\|\ell^t_{x} - \ell^{t-1}_{x}\|_{\infty} \right)\\
	  &\hspace{1cm} - \frac{1}{4\eta}\sum_{t\in [T-1]}\|\ell^{t}_{x} - \ell^{t-1}_{x}\|_{\infty}^{2}+\frac{1}{4\eta}  && \text{Lemma~\ref{lem:movement}}\\
	=&~O\left(\frac{ \log n}{\eta}\right) + \sum_{t\in [T-1]}\left(O(\eta^2) \cdot \|\ell^{t}_{x} - \ell^{t-1}_{x}\|_{\infty} - \frac{1}{4\eta}\cdot  \| \ell^{t}_{x} - \ell^{t-1}_{x}\|_{\infty}^{2} \right)\\
	\leq &~ O\left(\frac{ \log n}{\eta}\right) + T\cdot O(\eta^5) 
	=~ O\left(T^{1/6} \log^{5/6} n \right).\vspace{0.1cm}
	\end{align*}
This finishes the proof of the theorem. 
\end{proof}


\section{Lower Bounds for Hedge in Two-Player Games}\label{sec:lower-bound}

We prove lower bounds for regrets of players
  when they both run the vanilla Hedge algorithm. 
We show that even in games with two actions, vanilla Hedge
  cannot perform asymptotically better than its guaranteed regret bound of $\smash{O(\sqrt{T})}$  
  under the adversarial setting.
   
\begin{theorem}
	\label{thm:lower}
	Suppose two players run the vanilla Hedge algorithm to play a two-action game
	  with initial strategy $(0.4,0.6)$.
Then for any sufficiently large $T$ and any learning rate $\eta>0$, 
  there is a game such that at least one player has regret 
  $\smash{\Omega(\sqrt{T})}$ after $T'$ rounds for some $\smash{T'\in [T:T+\sqrt{T}]}$. 
\end{theorem}

\begin{remark}
Theorem \ref{thm:lower} shows that even if players have a good estimation
  about the number~of rounds to play (i.e., between $T$ and $\smash{T+\sqrt{T}}$),
  vanilla Hedge with any learning rate $\smash{\eta(T)>0}$ picked using $T$ cannot promise to achieve
  a regret bound that is 
  asymptotically lower than $\smash{O(\sqrt{T})}$
  for every round $\smash{T'\in [T:T+\sqrt{T}]}$.
We would like to point out that 
	  the use of $(0.4,0.6)$ as the initial strategy instead of the uniform distribution is not crucial
  but only to simplify the construction  and analysis.
\end{remark}
Let $T$ be a sufficiently large integer. 
We will use three games $G_i=(A,B_i)$, $i\in \{1,2,3\}$,
  to handle three cases of the learning rate $\eta$, where \vspace{0.1cm}
\begin{align*}
A = \left(\begin{matrix}1 &-1\\ -1 &1\end{matrix}\right),
\quad
B_1 = \left(\begin{matrix}-1 &1\\1  &-1\end{matrix}\right),
\quad
B_2 = \left(\begin{matrix}1 &1\\ 1 &1\end{matrix}\right)
\quad\text{and}\quad 
B_3 = \left(\begin{matrix}1 &-1\\ -1 &1\end{matrix}\right).\vspace{0.1cm}
\end{align*}
We use $G_2$ to handle the case when $\eta\le 64/{(c_0\sqrt{T})}$ (see Appendix \ref{sec:case2}) 
  where $c_0\in (0,1]$ is a constant introduced below in Lemma \ref{lem:kl-lower}.
We use $G_3$ to handle the case when $\eta\ge 3$ (see Appendix \ref{sec:case3}).
The most intriguing case is 
  when the learning rate $\eta$ is between $\smash{64/(c_0\sqrt{T})}$ and $3$.
For this case we use the Matching Pennies game $G_1=(A,B_1)$. 

Let $x^t$ and $y^t$ denote strategies played in round $t$ by the first and second players, respectively.
Let $\smash{x^{\star}=y^{\star}=(0.5,0.5)}$.
The proof for this case  relies on the following lemma, which shows that 
  the KL divergence between $\smash{(x^\star,y^\star)}$ and $\smash{(x^T,y^T)}$ after $T$ rounds is at least $\smash{\Omega( \sqrt{T}\eta))}$.

\begin{lemma}
		\label{lem:kl-lower}
Suppose players run vanilla Hedge  for $T$ rounds
  with $\eta:16/\sqrt{T}\le \eta\le 3$.
Then $$ {\DKL(x^{\star}\|x^T)+\DKL(y^{\star} \| y^T)\ge c_0 \sqrt{T}\eta},\quad
\text{for some constant $c_0\in (0, 1]$.}$$
	\end{lemma}


We are now ready to prove Theorem \ref{thm:lower} for the main case when $ 64/(c_0\sqrt{T})\le \eta\le 3$.
	\begin{proof}[Proof of Theorem~\ref{thm:lower} for the main case]
For convenience we let $x_t=x^t(1) $ (or $y_t=y^t(1)$) denote the probability of playing the first
  action in $x^t$ (or $y^t$, respectively).		
		We first describe the high level idea behind the proof. Since we know the KL divergence is at least $\smash{c_0 \sqrt {T}\eta }$ at time $T$ by Lemma~\ref{lem:kl-lower}, at least one of $\smash{x_T}$ and 
		$\smash{y_T}$ is extremely close to either $0$ or $1$.
	Assume without loss of generality that this is the case for $x_T$. 
As a result, the probability of the first player playing the first action
  will not change much for the next $\smash{\sqrt{T}}$ rounds. 
		Consequently, during the next $\smash{\sqrt{T}}$ rounds, 
		one of the players must keep losing and the other player will keep winning.
This can be used to show that one of the two players must have
  regret at least $\smash{\Omega(\sqrt{T})}$ at some point $T'$ between 
  $T$ and $\smash{T+\sqrt{T}}$. 
		 
		 To make this more formal, 
		   let $\ell_x^t$ (or $\ell_y^t$) denote the cost vector
		   of the first (or the second) player at round $t$ and 
		   define $\smash{L_x^t}$ and $\smash{L_y^t}$  
		   to be the total loss up to round $t$ of the two players:
		   $$ {L}_{x}^{t} = \sum_{\tau\in [t]}\langle x^\tau, \ell_x^\tau \rangle
		   \quad\text{and}\quad L^{t}_y = \sum_{\tau\in [t]}\langle y^\tau, \ell_y^\tau \rangle.$$ 
		 Since $G_1=(A, B_1)$ is zero-sum, we have 
		   $\langle x^\tau,\ell_x^\tau\rangle +\langle y^\tau,\ell_y^\tau\rangle=0$ and thus,
		 $ {L}^t_{x} + {L}^t_{y} = 0$. 
		 {Moreover, 
		 noting that the sum of two rows of $A$ is zero,}
		 the first player can always guarantee an overall loss of at most $0$ when 
		   playing the best fixed action in hindsight.
		 Therefore, $\smash{\regret_{t}^{x} \geq {L}^t_{x}}$ and similarly 
		 $\smash{\regret_{t}^{y} \geq {L}^t_{y}}$. Combining this with $L_x^t+L_y^t=0$, we have
		 		$$ {\max\Big\{ \regret_{t}^{x}, \hspace{0.05cm}\regret_{t}^{y}\Big\} \geq | {L}^t_{x}| = | {L}^t_{y}|.}$$
		To finish the proof, it suffices to show that  
		  \begin{equation}\label{eq:lower7} {\big|L_x^{T'}\big|=\big|L_y^{T'}\big|\ge \Omega(\sqrt{T})},\quad\text{for some 
		  $T'\in [T:T+\sqrt{T}]$}.\end{equation}

Let  $\smash{L=c_0\sqrt{T}/8\le \sqrt{T}}$.		
		We have from Lemma \ref{lem:kl-lower} that the KL divergence
		  is at least $c_0\sqrt{T}\eta$ (using 
		  $\smash{\eta\ge 64/(c_0\sqrt{T})>16/\sqrt{T}}$).
		  We assume without loss of generality that 
		  $\smash{\DKL(x^\star \|x^T)\ge c_0\sqrt{T}\eta/2}$.
		We further assume without loss of generality that the second term is larger: 
$$
\frac{1}{2}\cdot \log \frac{1}{2(1-x_T)}\ge \frac{c_0\sqrt{T}\eta}{4}.
$$ 
It follows that 
		$x_T$ is very close to $1$:
		 $\smash{x_T\ge 1- \exp(-c_0\sqrt{T}\eta/2)}$, and we use
  this to show that $\smash{x_{T+\tau}}$ remains close to $1$ for all $\smash{\tau\in [L]}$. To see this is the case, we note that
  \begin{align*}
		\frac{x_{T + \tau}}{1 - x_{T+\tau}} \geq \exp (-2\eta\tau)\cdot
		\frac{x_{T}}{1 - x_T} \geq \frac{1}{2}\cdot \exp\left(-2\eta L+\frac{c_0\sqrt{T}\eta}{2}\right)=\frac{1}{2}\cdot \exp\left(\frac{c_0\sqrt{T}\eta}{4}\right) \geq 3,
		\end{align*} 
where we used $\eta \ge 64/(c_0\sqrt{T})$ in the last inequality.
This implies $x_{T+\tau}\ge 3/4$ for all $\tau\in [L]$.

Now we turn our attention to the second player.
Given that $x_{T+\tau}\ge 3/4$ for all $\tau\in [L]$,
  $y_{T+\tau}$ keeps growing for all $\tau\in [L]$.
As a result there is an interval $I\subseteq [L]$ such that 
  (i) every $y_{T+\tau}$, $\tau\in I$, lies between $1/4$ and $3/4$;
  (ii) every $y_{T+\tau}$ before $I$ is smaller than $1/4$; and 
  (iii) every $y_{T+\tau}$ after $I$ is larger than $3/4$.
Using a similar argument, we show that $I$ cannot be too long.
Letting $\ell$ and $r$ be the left and right endpoints of $I$, we have 
$$
3\ge \frac{y_r}{1-y_r}\ge \exp\left(\frac{\eta(r-\ell)}{2}\right)\cdot \frac{y_\ell}{1-y_\ell}
\ge \exp\left(\frac{\eta(r-\ell)}{2}\right)\cdot \frac{1}{3}.
$$
As a result, we have $(r-\ell)\le 6/\eta\le (3/32)\cdot c_0\sqrt{T}$ and thus,
  either (i) or (ii) is of length at least $\Omega(L)$. 
We focus on the case when (ii) is long; the other case can be handled similarly.
 	
Summarizing what we have so far,
  there is an interval $\smash{J=[\alpha:\beta]\subseteq [L]}$ of length $\smash{\Omega(L)}$
  such that for every $\smash{\tau\in J}$, both $\smash{x_{T+\tau}}$ and $\smash{y_{T+\tau}}$ are at least $3/4$.
This implies that the total loss of the first player grows by $\Omega(1)$
  each round and thus,
$\smash{L_{x}^{T+\beta}-L_{x}^{T+\alpha}}\ge \Omega(L).$  
Therefore, either $|L_{x}^{T+\alpha}|\ge \Omega(L)$ or $\smash{|L_{x}^{T+\beta}|\ge \Omega(L)}$.
This finishes the proof of (\ref{eq:lower7}) using $\smash{L=\Omega(\sqrt{T})}$ and the proof of the theorem.
	\end{proof}

	
	

\section{Faster Convergence of Swap Regrets}
\label{sec:correlated}

Under the adversarial online model, Blum and Mansour \cite{blum2007external} gave 
  a black-box reduction showing that any algorithm that achieve good regrets 
  can be converted into an algorithm that achieves good swap regrets. 
In this section we show that if every player in a repeated game 
  runs their algorithm with optimistic Hedge as its core, 
  then the swap regret of each player  can be bounded from above by $\smash{O((n\log n)^{3/4}  (mT)^{1/4})}$, where $m$ is the number of players and $n$ is the 
  number of actions. 


\def\ALG{\texttt{ALG}}

We start with an overview on the reduction framework of~\cite{blum2007external}, which we will
  refer to as the BM algorithm. 
Let $S=[n]$ be the set of available actions.
Given an algorithm $\ALG$ that achieves good regrets,
  the 
  BM algorithm instantiates $n$ copies  $\ALG_1, \ldots, \ALG_{n}$ of $\ALG$ over $S$. 
At the beginning of each round $t=1,\ldots,T$, the BM algorithm receives 
  a distribution $q^t_i$ over $S$ from $\ALG_i$ for each $i\in [n]$,
  and plays $x^{t}$, which is the unique distribution over $S$
  that satisfies $x^t= x^tQ^{t}$, where $Q^{t}$
  is the $n\times n$ matrix with row vectors $q^t_1,\ldots,q^t_n$.
After receiving the loss vector $\ell^{t}$, the BM algorithm experiences a loss
  of $\langle x^t,\ell^t\rangle$ and distributes $x^{t}(i)\cdot \ell^{t}$ to $\ALG_i$
  as its loss vector in round $t$. 
  
 
We are now ready to state our main theorem of this section:

\begin{theorem}
\label{thm:correlated-equilibrium}
Suppose that every player in a repeated game runs
  the BM algorithm with optimistic Hedge as $\ALG$ and sets the learning rate of the latter
  to be $\eta = ({n\log n}/{(m^2T)} )^{1/4}$.
Then the swap regret of each player is
$\smash{O(( {n\log n} )^{3/4}\cdot (m^2T)^{1/4} )}$.
\end{theorem}

For convenience we refer to the BM algorithm with optimistic Hedge as
  \BMHedge\ in the rest of the section.
We first combine the analysis of \cite{blum2007external}
  for the BM algorithm and Lemma \ref{eq:opt-hedge-update}
  to obtain the following bound for the swap regret of 
  \BMHedge\ under the adversarial setting, in terms of the total path length of cost vectors
  the player's mixed strategies:
\begin{lemma}
\label{lem:correlated-equilibrium1}
Suppose that a player runs \emph{\BMHedge} with $\eta>0$ for $T$ rounds. Then  
\begin{align*}
    \sregret_{T} 
    \leq &~ \frac{2n\log n}{\eta} + 2\eta \left(\sum_{t=2}^T \| x^t -  x^{t-1}\|_{1}^{2} + 
    \sum_{t=1}^T \| \ell^t - \ell^{t-1}\|_{\infty}^2\right) ,\quad\text{where $\ell^0=\mathbf{0}$.}
\end{align*}
\end{lemma}

The proof can be found in Appendix \ref{sec:lem:correlated-equilibrium1}.
For the repeated game setting, we have for each $t\ge 2$,
$$\|\ell^t_i-\ell^{t-1}_i\|_\infty\le \|\bx_{-i}^t-\bx_{-i}^{t-1}\|_1
\le \sum_{j\ne i} \|\bx_j^t-\bx_j^{t-1}\|_1$$
where the last inequality used the fact that both $\bx_{-i}^t$ and $\bx_{-i}^{t-1}$
  are product distributions.
Combining it with Lemma \ref{lem:correlated-equilibrium1}, we 
  can bound the swap regret of each player $i\in [m]$ in the game by 
\begin{align}\label{mainbound}
    \sregret^i_{T} 
    \leq &~ \frac{2n\log n}{\eta} + 2\eta + 2\eta m 
    \sum_{j\in [m]} \sum_{t=2}^T \|x^t_j-x^{t-1}_j\|_1^2.
\end{align}
We prove the following main technical lemma in the rest of the section, which states that 
  the mixed strategy $x^t$ produced by $\BMHedge$ under the adversarial setting
  moves very slowly
  (by at most $O(\eta)$ in $\ell_1$-distance each round).
Theorem \ref{thm:correlated-equilibrium}  follows by combining Lemma \ref{lem:correlated-equilibrium1} and \ref{main-technical-lemma}.  
  
\begin{lemma}\label{main-technical-lemma}
Suppose that a player runs $\emph{\BMHedge}$ with rate $\eta:0<\eta\le 1/6$ under the adversarial setting.
Then we have $\|x^t-x^{t-1}\|_1\le O(\eta)$ for all $t\ge 2$.  
\end{lemma}  
  
\begin{proof}[Proof of Theorem \ref{thm:correlated-equilibrium} Assuming Lemma \ref{main-technical-lemma}]
Let $\eta = (n\log n)^{1/4}(m^2T)^{-1/4}$.
For the special case when $\eta>1/6$,  the swap regret of
  each player is trivially at most $\smash{T=O( (n\log n)^{3/4}\cdot(m^2T)^{1/4} )}$.
Assuming $\eta\le 1/6$, 
  by Lemma~\ref{lem:correlated-equilibrium1} we have from (\ref{mainbound}) that
	\begin{align*}
	\sregret^i_{T} 
	 \leq  \frac{2n\log n}{\eta} +2\eta + 2\eta m^2 T\cdot O(\eta^2) 
	 =O\left(  (n\log n)^{3/4}\cdot  (m^2T)^{1/4} \right).
	\end{align*}
This finishes the proof of the theorem.
\end{proof}
  

\begin{figure}[t!]
	\begin{minipage}{\textwidth}
		\centering
		\begin{align*}
		Q =  
		\begin{pmatrix}
		1 - \eps & \eps\\
		\eps' & 1 - \eps'
		\end{pmatrix}
		  \ \ 
		x =  
		\left(\frac{1}{k+1}\ \ \ 
		\frac{k}{k+1}\right)
		 \ \ \text{vs}\ \ 
		Q = 
		\begin{pmatrix}
		1 - \eps' & \eps'\\
		\eps & 1 - \eps
		\end{pmatrix}
	  \  \
		x = 
		\left(\frac{k}{k+1}\ \ \ 
		\frac{1}{k+1}\right) 
		\end{align*}
	\end{minipage}\vspace{0.1cm}
	\caption{Let $\eps'=\eps/k$. Additive perturbations may change the stationary distribution dramatically.}
	\label{fig:correlated1}
\end{figure}

The proof of Lemma \ref{main-technical-lemma} can be found in Appendix \ref{lemma53}.
Here we give a high-level description of its proof.
Given that $\BMHedge$ runs $n$ copies of optimistic Hedge with rate $\eta$,
  we know that mixed strategies proposed by each $\ALG_i$ move very slowly:
  $\smash{\|q_i^t-q_i^{t-1}\|_1\le O(\eta)}$.
However, it is not clear whether this translates into a similar property for $x^t$
  since the latter is obtained by solving
  $x^t=x^tQ^t $.
Equivalently, $x^t$ can be viewed as the stationary distribution 
  of the Markov chain $Q^t$
  composed by strategies of each individual expert $\ALG_i$, and its dependency
  on $Q^t$ is highly nonlinear.
While there is a vast literature on 
  the perturbation analysis of Markov chains,
  many results require additional assumptions on the 
  underlying Markov chain (e.g. bounded eigenvalue gap) and 
  are not well suited for our setting here.
Indeed, it is easy to come up with examples showing that the stationary distrbution is extremely sensitive to small \emph{additive} perturbations (see Figure~\ref{fig:correlated1}).
As a result one cannot hope to prove Lemma \ref{main-technical-lemma} based on  
  the property $\smash{\|q_i^t-q_i^{t-1}\|_1\le O(\eta)}$ only.

We circumvent this difficulty by noting that optimistic Hedge only incurs small 
{\em multiplicative} perturbations on the Markov chain (see Claim~\ref{lem:optimistic-perturbation}), i.e., each entry 
  of $\smash{Q^{t}}$ differs from the corresponding entry of $\smash{Q^{t-1}}$ by no more than a   small multiplicative factor of the latter.
We present in Lemma~\ref{lem:perturbation} an analysis on stationary distributions of Markov chains under multiplicative perturbations, based on the classical
  Markov chain tree theorem, 
  and then use it to prove Lemma \ref{main-technical-lemma}.

We further prove that one can design a wrapper for 
  $\BMHedge$ that is robust against adversarial opponents:
\begin{corollary}
	\label{cor:robus-adv}
There is an algorithm $\emph{\BMHedge}^*$ with the following guarantee.
If all players run $\emph{\BMHedge}^*$, then the swap regret of each individual is $\smash{\tilde{O}(n^{3/4}(m^2 T)^{1/4} )}$; if the player is facing adversaries, then the swap regret is still at most $\smash{\tilde{O} ((nT)^{1/2} +n^{3/4}(m^2 T)^{1/4}  )}$.
\end{corollary}

In the appendix we give two more extensions to our results
  on swap regrets. 
\begin{flushleft}\begin{enumerate}
\item In Appendix \ref{sec:extension-phi}, we show that incorporating optimistic Hedge into a folklore algorithm  from~\cite{cesa2006prediction} can also achieve faster convergence of swap regrets, with a slightly worse dependence on $n$. Interestingly, our analysis of this algorithm also crucially relies on the perturbation analysis of stationary distributions of Markov chains.
\item In Appendix \ref{sec:price},
 we study the convergence to the approximately optimal social welfare (following the definition in~\cite{foster2016learning}) with no-swap regret algorithms, and prove that $O(1/T)$ holds for a wide range of no-swap regret algorithms. 
\end{enumerate}\end{flushleft}


\section{Discussion}
\label{sec:discussion}

In this paper, we studied the convergence rate of regrets of the Hedge algorithm and its optimistic variant in 
  two-player games. 
We obtained a strict separation between vanilla Hedge and optimistic Hedge, i.e., $ \smash{1/\sqrt{T}} $ vs. $ \smash{1/T^{5/6}} $. 
We also initiated the study on algorithms with faster convergence rates
  of swap regrets in general multiplayer games and obtained an algorithm
    with average regret $\smash{O(m^{1/2}(n\log n/T)^{3/4})}$ , improving over the classic result of Blum and Mansour~\cite{blum2007external}.

Our work led
  to several interesting future directions:
\begin{flushleft}\begin{itemize}
	\item Our faster convergence result for optimistic Hedge currently only works for two-player games. Can we extend it to multiplayer games? Second, what is the optimal convergence rate for optimistic Hedge and other no-regret algorithms? even for two-player games? 
	
	\item Regarding swap regrets, it is easy to generalize the result in Section~\ref{sec:correlated} to any algorithm that (1) satisfies the RVU property and (2) makes only multiplicative changes on strategies each iteration. These include optimistic Hedge and optimistic multiplicative weights. However, 
	our current analysis does not apply to
	general optimistic Mirror Descent or Follow the Regularized Leader. 
Can we still prove faster convergence of swap regrets via the reduction of \cite{blum2007external}
  without requiring (2) on the regret minimization algorithm?
  or does there exist some natural gap between these algorithms and optimistic Hedge\hspace{0.05cm}/\hspace{0.05cm}multiplicative weights?
	
	\item For our result in Appendix \ref{sec:price} on the convergence to the approximately 
	  optimal social welfare, can this fast convergence result be extended to the (exact) optimal social welfare setting (follow the definition in~\cite{syrgkanis2015fast})?
	
	
	\item  Can we achieve similar convergence rates under partial information models? such as those
	  considered in \cite{rakhlin2013optimization, foster2016learning,wei2018more}.
\end{itemize}\end{flushleft}

\section*{Acknowledgement}
Binghui Peng would thank Christos H. Papadimitriou for useful discussions.

\newpage
\bibliographystyle{plain}
\bibliography{ref}

\begin{thebibliography}{10}

\bibitem{anantharam1989proof}
Venkat Anantharam and Pantelis Tsoucas.
\newblock A proof of the markov chain tree theorem.
\newblock {\em Statistics \& Probability Letters}, 8(2):189--192, 1989.

\bibitem{arora2012multiplicative}
Sanjeev Arora, Elad Hazan, and Satyen Kale.
\newblock The multiplicative weights update method: a meta-algorithm and
  applications.
\newblock {\em Theory of Computing}, 8(1):121--164, 2012.

\bibitem{Aumann74}
R.J. Aumann.
\newblock Subjectivity and correlation in randomized strategies.
\newblock {\em Journal of Mathematical Economics}, 1:67--96, 1974.

\bibitem{bailey2019fast}
James Bailey and Georgios Piliouras.
\newblock Fast and furious learning in zero-sum games: vanishing regret with
  non-vanishing step sizes.
\newblock In {\em Advances in Neural Information Processing Systems}, pages
  12977--12987, 2019.

\bibitem{bailey2018multiplicative}
James~P Bailey and Georgios Piliouras.
\newblock Multiplicative weights update in zero-sum games.
\newblock In {\em Proceedings of the 2018 ACM Conference on Economics and
  Computation}, pages 321--338, 2018.

\bibitem{blum2007external}
Avrim Blum and Yishay Mansour.
\newblock From external to internal regret.
\newblock {\em Journal of Machine Learning Research}, 8(Jun):1307--1324, 2007.

\bibitem{cesa2006prediction}
Nicolo Cesa-Bianchi and G{\'a}bor Lugosi.
\newblock {\em Prediction, learning, and games}.
\newblock Cambridge university press, 2006.

\bibitem{cheung2019vortices}
Yun~Kuen Cheung and Georgios Piliouras.
\newblock Vortices instead of equilibria in minmax optimization: Chaos and
  butterfly effects of online learning in zero-sum games.
\newblock In {\em Conference on Learning Theory}, pages 807--834, 2019.

\bibitem{daskalakis2011near}
Constantinos Daskalakis, Alan Deckelbaum, and Anthony Kim.
\newblock Near-optimal no-regret algorithms for zero-sum games.
\newblock In {\em Proceedings of the twenty-second annual ACM-SIAM symposium on
  Discrete Algorithms}, pages 235--254. SIAM, 2011.

\bibitem{daskalakis2018training}
Constantinos Daskalakis, Andrew Ilyas, Vasilis Syrgkanis, and Haoyang Zeng.
\newblock Training gans with optimism.
\newblock In {\em International Conference on Learning Representations}, 2018.

\bibitem{daskalakis2018last}
Constantinos Daskalakis and Ioannis Panageas.
\newblock Last-iterate convergence: Zero-sum games and constrained min-max
  optimization.
\newblock {\em arXiv preprint arXiv:1807.04252}, 2018.

\bibitem{foster1999regret}
Dean~P Foster and Rakesh Vohra.
\newblock Regret in the on-line decision problem.
\newblock {\em Games and Economic Behavior}, 29(1-2):7--35, 1999.

\bibitem{foster1997calibrated}
Dean~P Foster and Rakesh~V Vohra.
\newblock Calibrated learning and correlated equilibrium.
\newblock {\em Games and Economic Behavior}, 21(1-2):40, 1997.

\bibitem{foster2016learning}
Dylan~J Foster, Zhiyuan Li, Thodoris Lykouris, Karthik Sridharan, and Eva
  Tardos.
\newblock Learning in games: Robustness of fast convergence.
\newblock In {\em Advances in Neural Information Processing Systems}, pages
  4734--4742, 2016.

\bibitem{freund1996game}
Yoav Freund and Robert~E Schapire.
\newblock Game theory, on-line prediction and boosting.
\newblock In {\em Proceedings of the ninth annual conference on Computational
  learning theory}, pages 325--332, 1996.

\bibitem{Hedge97}
Yoav Freund and Robert~E. Schapire.
\newblock {\em J. Comput. System Sci.}, 55(1):119--139, 1997.

\bibitem{greenwald2008more}
Amy Greenwald, Zheng Li, and Warren Schudy.
\newblock More efficient internal-regret-minimizing algorithms.
\newblock In {\em COLT}, pages 239--250, 2008.

\bibitem{hart2000simple}
Sergiu Hart and Andreu Mas-Colell.
\newblock A simple adaptive procedure leading to correlated equilibrium.
\newblock {\em Econometrica}, 68(5):1127--1150, 2000.

\bibitem{hazan2016introduction}
Elad Hazan.
\newblock Introduction to online convex optimization.
\newblock {\em Foundations and Trends in Optimization}, 2(3-4):157--325, 2016.

\bibitem{kalai2005efficient}
Adam Kalai and Santosh Vempala.
\newblock Efficient algorithms for online decision problems.
\newblock {\em Journal of Computer and System Sciences}, 71(3):291--307, 2005.

\bibitem{mertikopoulos2018cycles}
Panayotis Mertikopoulos, Christos Papadimitriou, and Georgios Piliouras.
\newblock Cycles in adversarial regularized learning.
\newblock In {\em Proceedings of the Twenty-Ninth Annual ACM-SIAM Symposium on
  Discrete Algorithms}, pages 2703--2717. SIAM, 2018.

\bibitem{nisan2007algorithmic}
Noam Nisan, Tim Roughgarden, Eva Tardos, and Vijay~V Vazirani.
\newblock {\em Algorithmic Game Theory}.
\newblock Cambridge University Press, 2007.

\bibitem{rakhlin2013online}
Alexander Rakhlin and Karthik Sridharan.
\newblock Online learning with predictable sequences.
\newblock In {\em Conference on Learning Theory}, pages 993--1019, 2013.

\bibitem{rakhlin2013optimization}
Sasha Rakhlin and Karthik Sridharan.
\newblock Optimization, learning, and games with predictable sequences.
\newblock In {\em Advances in Neural Information Processing Systems}, pages
  3066--3074, 2013.

\bibitem{roughgarden2015intrinsic}
Tim Roughgarden.
\newblock Intrinsic robustness of the price of anarchy.
\newblock {\em Journal of the ACM (JACM)}, 62(5):1--42, 2015.

\bibitem{shalev2011online}
Shai Shalev-Shwartz et~al.
\newblock Online learning and online convex optimization.
\newblock {\em Foundations and trends in Machine Learning}, 4(2):107--194,
  2011.

\bibitem{syrgkanis2015fast}
Vasilis Syrgkanis, Alekh Agarwal, Haipeng Luo, and Robert~E Schapire.
\newblock Fast convergence of regularized learning in games.
\newblock In {\em Advances in Neural Information Processing Systems}, pages
  2989--2997, 2015.

\bibitem{wang2018acceleration}
Jun-Kun Wang and Jacob~D Abernethy.
\newblock Acceleration through optimistic no-regret dynamics.
\newblock In {\em Advances in Neural Information Processing Systems}, pages
  3824--3834, 2018.

\bibitem{wei2018more}
Chen-Yu Wei and Haipeng Luo.
\newblock More adaptive algorithms for adversarial bandits.
\newblock In {\em Conference On Learning Theory}, pages 1263--1291, 2018.

\bibitem{zinkevich2003online}
Martin Zinkevich.
\newblock Online convex programming and generalized infinitesimal gradient
  ascent.
\newblock In {\em Proceedings of the 20th international conference on machine
  learning (icml-03)}, pages 928--936, 2003.

\end{thebibliography}

\appendix

\newpage

\section{Missing proof from Section~\ref{sec:coarse}}\label{sec:movement}
{\bf Proof of Lemma \ref{lem:movement}\ \ }
For each $t \in [2:T]$, we apply Pinsker's inequality to have
\begin{align}
\frac{1}{2}\cdot \|x^t - x^{t-1}\|_1^{2} 
\leq &~  \DKL(x^{t-1}\| x^{t})
=~ \sum_{i\in [n]}x^{t-1}(i)\cdot \log \left(\frac{x^{t-1}(i)}{x^t(i)}\right)\notag \\
= &~ \sum_{i\in [n]}x^{t-1}(i)\cdot \log \left(\sum_{j\in [n]}\exp\left(-\eta\big(2\ell^{t-1}(j) - \ell^{t-2}(j)\big) \right)\cdot x^{t-1}(j) \right) \notag \\
&~~~~~~~\ \ \ \ +  \sum_{i\in [n]}x^{t-1}(i)\cdot \eta \big(2\ell^{t-1}(i) - \ell^{t-2}(i)\big)\notag \\
= &~ \log \left(\sum_{j\in [n]}\exp\left(-\eta\big(2\ell^{t-1}(j) - \ell^{t-2}(j)\big)\right)\cdot x^{t-1}(j) \right) + \eta \langle x^{t-1}, 2\ell^{t-1} - \ell^{t-2} \rangle\notag\\
\triangleq &~ \Phi_t +\eta \langle x^{t-1}, 2\ell^{t-1} - \ell^{t-2} \rangle \label{eq:coarse1},
\end{align}
where we recall $\ell^{0} = \mathbf{0}$. 
The third step follows from the updating rule of optimistic Hedge. 
Letting $L^t=\sum_{i\in [t]} \ell^i$,
next we use induction to prove the following claim for each $k=1,\ldots,T$: 
\begin{align}
\label{eq:potential}
\sum_{t\in [k]}\Phi_t = \log \left(\sum_{j\in [n]}x^1(j)\cdot  \exp\left(-\eta L^{k-1}(j) -\eta \ell^{k-1}(j) \right) \right).
\end{align}
The base case holds trivially, as $\Phi_1 = 0$. Suppose the above holds for $k$. Then for $k+1$ we have
\begin{align*}
\sum_{t=1}^{k+1}\Phi_t 
=&~ \sum_{t=1}^{k}\Phi_{t}+ \Phi_{k + 1}\\
=&~\log \left(\sum_{j\in [n]}x^1(j)\cdot \exp\left(-\eta L^{k-1}(j) -\eta \ell^{k-1}(j) \right) \right) + \log \left(\sum_{i\in [n]}\exp\left(-\eta\big(2\ell^{k}(i) - \ell^{k-1}(i)\big) \right)\cdot x^{k}(i) \right)\\
=&~ \log \left(\left(\sum_{i\in [n]}\exp\left(-\eta\big(2\ell^{k}(i) - \ell^{k-1}(i)\big) \right)\cdot x^{k}(i)\right) \cdot\left( \sum_{j\in [n]}x^1(j)\cdot \exp\left(-\eta L^{k-1}(j) -\eta \ell^{k-1}(j)\right) \right) \right)\\
=&~\log \left( \sum_{i\in [n]}\exp\left(-\eta\big(2\ell^{k}(i) - \ell^{k-1}(i)\big) \right) \cdot x^1(i) \cdot\exp\left(-\eta L^{k-1}(i) - \eta \ell^{k-1}(i)\right) \right)\\
=&~\log \left(\sum_{i\in [n]}x^1(i)\cdot \exp\left(-\eta L^{k}(i) -\eta \ell^{k}(i) \right) \right),
\end{align*}
where the third step follows from 
\begin{align*}
x^k(i) = \frac{x^1(i)\cdot \exp\left(-\eta L^{k-1}(i) -\eta \ell^{k-1}(i)\right)}{\sum_{j\in [n]}x^1(j)\cdot \exp\left(-\eta L^{k-1}(j) - \eta \ell^{k-1}(j)\right)}.
\end{align*}
Now we have (recall that $\Phi_1=0$)
\begin{align*}
\frac{1}{2\ln 2}\sum_{t\in [2:T]} \|x^t - x^{t-1}\|_1^{2}  
\leq &~ \sum_{t\in [2:T]}\Big(\Phi_{t} +  \eta \langle x^{t-1}, 2\ell^{t-1} - \ell^{t-2} \rangle\Big)  \\
= &~  \log \left(\sum_{j\in [n]}\frac{1}{n}\cdot \exp\left(-\eta L^{T-1}(j) -\eta \ell^{T-1}(j) \right) \right) + \sum_{t\in [2:T]}\eta \langle x^{t-1}, 2\ell^{t-1} - \ell^{t-2} \rangle\\
\leq  &~-\min_{j \in [n]}\Big(\eta L^{T-1}(j) +\eta \ell^{T-1}(j) \Big) + \sum_{t\in [2:T]}\eta \langle x^{t-1}, 2\ell^{t-1} - \ell^{t-2} \rangle\\
\leq & -\eta\min_{j\in [n]} L^{T-1}(j) + \eta\sum_{t\in [T-1]}\langle x^{t}, \ell^{t} \rangle +  \eta\sum_{t\in [T-1]}\langle x^{t }, \ell^{t } - \ell^{t-1} \rangle\\
\leq &~ \eta \left(\frac{2\log n}{\eta} + \eta \sum_{t\in [T-1]}\|\ell^t - \ell^{t-1}\|_{\infty}^2 \right) + \eta\sum_{t\in [T-1]}\langle x^{t }, \ell^{t } - \ell^{t-1}\rangle\\
\leq&~2\log n + \eta^2\sum_{t\in [T-1]}\|\ell^t - \ell^{t-1}\|_{\infty}^2 + \eta \sum_{t\in [T-1]}\|\ell^{t } - \ell^{t-1}\|_{\infty}\\
\leq&~ 2\log n + { (\eta+\eta^2) \sum_{t\in [T-1]}\|\ell^{t} - \ell^{t-1}\|_{\infty}}.
\end{align*}
The first step follows from Eq.~\eqref{eq:coarse1} and the second step follows from Eq.~\eqref{eq:potential}. The fifth step follows from Lemma~\ref{lem:optimisticRegret}. This finishes the proof of the lemma. 

\section{Missing proof from Section~\ref{sec:lower-bound}}
\subsection{Case when the learning rate is small}\label{sec:case2}

We handle the case when $\eta\le 64/(c_0\sqrt{T})=O(1/\sqrt{T})$ with the following lemma:

\begin{lemma}
	\label{lem:learning-rate-small}
	Suppose both players run vanilla Hedge on game $G_2=(A,B_2)$ 
	with learning rate $\smash{\eta = O(1/\sqrt{T})}$.
	Then the regret of the first player is at least $\smash{\Omega(\sqrt{T})}$ after $T$ rounds.
\end{lemma}
\begin{proof}
	The loss of player 2 is invariant to the strategy of player 1. Thus her strategy stays at $(0.4, 0.6)$. Hence, for any $t \in [T]$, the loss for player 1 is always $\ell=(-0.2, 0.2)$ and we have\vspace{0.1cm}
	\begin{align*}
	x^t(1) &= \frac{0.4\cdot \exp(0.2\eta t)}{0.4\cdot \exp(0.2\eta t) + 0.6\cdot\exp(-0.2\eta t)} \text{\quad   and   }\\[0.6ex]	x^t(2) &= \frac{0.6\cdot \exp(-0.2\eta t)}{0.4\cdot \exp(0.2\eta t) + 0.6\cdot\exp(-0.2\eta t)}.\\[-2.2 ex]
	\end{align*}
	One can verify that when $t \leq  {1}/{2\eta}$, we have $x^t(1) \leq 0.5 \leq x^t(2)$. Therefore,  the regret is
	\begin{align*}
	\regret_{T}^{x} = \sum_{t\in [T]}\langle x^t, \ell\rangle - \sum_{t\in [T]}\ell(1) \geq \sum_{t=1}^{1/2\eta}\langle x^t, \ell\rangle - \sum_{t=1}^{1/2\eta}\ell(1) \geq  0 + \frac{1}{2\eta}\cdot  0.2 = \Omega(\sqrt{T}).
	\end{align*}
	Thus we complete the proof.
\end{proof}

\subsection{Case when the learning rate is large}\label{sec:case3}

We next work on the case when $\eta\ge 3$. 
Recall that we write $x_t = x^{t}(1)$ and $y_t = y^t(1)$. 

\begin{lemma}
	\label{lem:learning-rate-large}
Suppose both players run vanilla Hedge on game $G_3=(A, B_3)$ with 
  learning rate $\eta\ge 3$
Then the regret of the first player is at least $\Omega(T)$ after $T$ rounds.
\end{lemma}
\begin{proof}
	Intuitively, $(A, B_3)$ is a cooperation game, and it is beneficial for both players if they choose to cooperate on one single action (by playing either $(1,2)$ or $(2,1)$). 
	However, when the learning rate is too large, they actually mismatch in every iterations. 
	Formally, we have\vspace{0.1cm}
	\begin{align*}
	x_{t + 1} = &~\frac{x_t\cdot \exp(\eta(1 - 2y_t))}{x_t\cdot \exp(\eta(1 -2y_t)) + (1 - x_t)\cdot \exp(\eta(2y_t -1 )) } \\[0.4ex]
	=&~ \frac{x_t\cdot \exp(\eta(1 - 2x_t))}{x_t\cdot \exp(\eta(1 -2x_t)) + (1 - x_t)
	\cdot \exp(\eta(2x_t -1 )) }.\\[-2.2 ex]
	\end{align*}
	The second step follows from $x_t = y_t$ for all $t$
	because $A=B_3$ in the game. 
Motivated by this, we define a sequence $a_0,a_1,\ldots$ where $a_0 = x_0 = 0.4$ and 
	\begin{align*}
a_{t + 1} =  \frac{(1 - a_t)\cdot \exp(\eta(2a_t -1))}{a_t\cdot \exp(\eta(1 - 2a_t)) + (1 - a_t)\cdot \exp(\eta(2a_t -1 ))}, \quad\text{for each $t\ge 0$.}
	\end{align*}
Then $a_t = x_t$ if $t$ is even and $a_t = 1- x_t$ when $t$ is odd. 
Furthermore, by Claim~\ref{claim:tech1} below, we have 
$	\eta\exp(-2\eta)\leq a_t \leq 0.4$	
  for all $t$ when $\eta \geq 3$. Hence, we have
	\begin{align*}
	\regret_{T}^{x} \geq \sum_{t\in [T]}\langle x^t, \ell_x^t\rangle = \sum_{t\in [T]}(2x_t - 1)^{2} = \sum_{t\in [T]}(2a_t - 1)^{2} \geq  \Omega(T).
	\end{align*}
This finishes the proof of the lemma.
\end{proof}


\begin{claim}
	\label{claim:tech1}
When $\eta\ge 3$, we have \hspace{0.04cm}$\eta\exp(-2\eta)\leq a_t \leq 0.4$\ \hspace{0.06cm}for all $t\ge 0$. 
\end{claim}
\begin{proof}
	We prove by induction on $t$. The base case holds trivially for $t = 0$. Suppose the inequality holds up to $t$. Then for $t + 1$, we have
	\begin{align*}
	\frac{a_{t + 1}}{1 - a_{t + 1}} = \frac{ 1 - a_t }{a_t}\cdot \exp\big(\eta(4a_t - 2)\big) \triangleq f(a_t).
	\end{align*}
	By simple calculation, we know that $f(a_t)$ takes maximium at $\eta\exp(-2\eta)$ or $0.4$.
	Thus,	\begin{align*}
	\frac{a_{t + 1}}{1 - a_{t + 1}} \leq \max\Big\{f (0.4 ), f(\eta\exp(-2\eta))\Big\} \leq \frac{2}{3},
	\end{align*}
	which implies that $a_{t + 1} \leq  0.4.$
	The second step above follows from
	\begin{align*}
	f (0.4 ) = \frac{3}{2}\cdot \exp( -0.4\eta)\leq \frac{2}{3}, 
	\end{align*}
	using $\eta\ge 3$ and
	\begin{align*}
	f\big(\eta\exp(-2\eta)\big) \leq \frac{1}{\eta}\exp(2\eta) \cdot \exp\big(4\eta^{2}\exp(-2\eta) - 2\eta\big) = \frac{1}{\eta} \cdot \exp\big(4\eta^2 \exp(-2\eta)\big) \leq \frac{2}{3}.
	\end{align*}
	Moreover, $f(a_t)$ takes minimum at the smaller solution $a$ of $4\eta a (1-a) = 1$. Thus,
	\begin{align*}
	\frac{a_{t + 1}}{1 - a_{t + 1}} \geq \frac{1 - a}{a}\cdot \exp\big(\eta (4a - 2)\big) \geq  \frac{4}{3}\cdot \eta\exp(-2\eta),   
	\end{align*}
where the second step used $\exp (\eta (4a - 2))\geq \exp(-2\eta)$, $a \leq 1/{2\eta}$ and $a\le 1/3$.
This shows that $a_{t + 1} \geq \eta\exp(-2\eta)$ using $\eta\ge 3$, and finishes the induction.
\end{proof}

\subsection{Proof of Lemma \ref{lem:kl-lower}}\label{sec:KLdivergence}

Note that the Matching Pennies game $G_1=(A,B_1)$ is zero-sum.
It is known (see~\cite{bailey2018multiplicative}) that the KL divergence of vanilla Hedge in
  zero-sum games is strictly increasing. We give a careful analysis on its increment each round 
  when playing  $G_1$. (Recall that $x^\star=y^{\star}=(0.5,0.5)$.)
\begin{lemma}
	\label{lem:kl}
Suppose both players run vanilla Hedge with $\eta\le 3$ on $G_1$.
Then for each $t\ge 0$, 
	\begin{align*}
&\DKL(x^\star \| x^{t+1})+\DKL(y^\star \| y^{t+1})-\big(\DKL(x^\star \| x^{t })+\DKL(y^\star \| y^{t })\big)\\[0.2ex]
	 &\hspace{1.5cm}\geq~ e^{-7}\eta^2 x_t(1-x_t)(2y_t-1)^2 +e^{-7}\eta^2 y_t(1-y_t)(2x_t-1)^2.  
	\end{align*}
\end{lemma}
\begin{proof}
	Focusing on the first player, we have
	\begin{align}
	&\DKL(x^{\star}\| x^{t+1}) - \DKL(x^{\star}\| x^{t}) \notag\\[0.3ex]
	&=~ \sum_{i\in [2]}x^{\star}(i)\cdot \log \left(\frac{x^{\star}(i)}{x^{t+1}(i)}\right) - \sum_{i\in [2]}x^{\star}(i)\cdot \log\left( \frac{x^{\star}(i)}{x^{t}(i)}\right)\notag\\
	&=~ \sum_{i\in [2]}x^{\star}(i)\cdot  \log \left(\frac{x^t(i)}{x^{t+1}(i)}\right)\notag\\
	&=~\sum_{i\in [2]}x^{\star}(i)\cdot \eta \ell^t(i) + \sum_{i\in [2]}x^{\star}(i)\cdot \log \left(\sum_{j\in [2]} x^{t}(j)\cdot \exp(-\eta \ell^t(j))\right)\notag\\
	 &=~\log \left(\sum_{j\in [2]} x^{t}(j)\cdot \exp(-\eta \ell^t(j))\right)\notag\\[0.4ex]
	 &=~\log \Big(x_t\cdot  \exp(-\eta (2y_t - 1)) + (1 - x_t)\cdot \exp(-\eta (1 - 2y_t))\Big)\notag\\
	 &\geq~x_t \cdot (-\eta (2y_t - 1)) + (1 - x_t) \cdot (-\eta (1 - 2y_t)) + \frac{1}{2e^6}x_t (1 - x_t)\left(e^{-\eta (2y_t - 1)} - e^{-\eta (1 - 2y_t)}\right)^2\notag\\
	 &\geq~ \eta(2y_t-1)(1-2x_t) +  e^{-7}\eta^2 x_t (1 - x_t)(2y_t - 1)^2.\label{eq:lower3}
	\end{align}
	The third step follows from the updating rule of vanilla Hedge. The fourth step uses $x^{\star}(1) = x^{\star}(2) = 0.5$ and $\ell^{t}(1) + \ell^{t}(2) = (2y_t - 1) + (1 - 2y_t) = 0$. The sixth step uses the fact that   $f(x) = -\log x$ is $ {e^{-6}}$-strongly convex on $(0, e^3)$.
	Similarly, we can prove
	\begin{align}
	\DKL(y^{\star}\| y^{t+1}) - \DKL(y^{\star}\| y^{t}) 
	\geq \eta(2x_t-1)(2y_t-1) +   {e^{-7}}\eta^2 y_t (1 - y_t)(2x_t - 1)^2.\label{eq:lower4}
	\end{align}
	The lemma follows by combining (\ref{eq:lower3}) and (\ref{eq:lower4}). 
\end{proof}
	
We are now ready to prove Lemma \ref{lem:kl-lower}.	
	\begin{proof}[Proof of Lemma \ref{lem:kl-lower}]
		We first prove that within $O(1/{\eta^2})$ steps, 
		  the KL divergence $\DKL(x^\star\|x^t)+\DKL(y^{\star}\| y^t)$ 
		  becomes at least $20$. The proof follows directly from Lemma~\ref{lem:kl}, as for any $t$ with $\DKL(x^\star\|x^t)+\DKL(y^{\star}\| y^t)\le 20$, we have
		\begin{align}
&\DKL(x^\star \| x^{t+1})+\DKL(y^\star \| y^{t+1})-\big(\DKL(x^\star \| x^{t })+\DKL(y^\star \| y^{t })\big)\notag\\[0.2ex]
	 &\hspace{1.5cm}\geq~ e^{-7}\eta^2 x_t(1-x_t)(2y_t-1)^2 +e^{-7}\eta^2 y_t(1-y_t)(2x_t-1)^2\ge \Omega(\eta^2).\label{hehehaha}
		\end{align}
		The second step follows from the fact that both
		  $x_t$ and $y_t$ are bounded away from $0$ and $1$ given
		  the divergence at $t$ is at most $20$;
		  it also used 
		  $\max\{|2x_t - 1|, |2y_t - 1|\} \geq 0.2$ given that the divergence
		  is strictly increasing. 
		  
Let $T_0=O(1/\eta^2)$ be the first time when the divergence becomes at least $20$.
If $T/2\le T_0$, it follows from (\ref{hehehaha}) that 
  the divergence at $T$ is  $\smash{\Omega(T\eta^2)=\Omega(\sqrt{T}\eta)}$ using the
  assumption that $\smash{\eta\ge 16/\sqrt{T}}$.
So we focus on the case $T_0\le T/2$ and thus, $T=T_0+L$ with $L\ge T/2$.
We prove

\begin{claim} 
			\label{claim:kl}
		At round $t=T_0+\tau^2$, the KL divergence has
		  $\DKL(x^{\star}\| x^t)+\DKL(y^{\star}\|y^t)\ge 10^{-10}\tau\eta$. 
		\end{claim}
		  
Setting $\tau=\sqrt{T/2}$ so that $T_0+\tau^2\le T$, we have
		  $$
		  \DKL(x^{\star}\| x^T)+\DKL(y^{\star}\|y^T)\ge \Omega(\sqrt{T}\eta),
		  $$
and this finishes the proof of the lemma.\end{proof}
\begin{proof}[Proof of Claim \ref{claim:kl}]				
		We proceed to use induction on $\tau$.
		The cases with $\tau\le 16/\eta$ holds trivially as the KL divergence at $T_0$ is already at least 20. For the induction step, suppose the claim holds up to $k$
		for some $k\ge 64/\eta$ at time
		$t_0=T_0+k^2$. 
		We show that at time $T_0+(k + 1)^2$ the KL divergence is at least $\smash{10^{-10}(k + 1)\eta}$.
		Without loss of generality, we assume that 
		  $x_{t_0}, y_{t_0} \geq 0.5$; the other three cases 
		  can be handled similarly. 
		In this region,  $x_t$ with $t=t_0+1,\ldots$ will keep decreasing and $y_t$ 
		will keep increasing, until the moment when $x_t$ drops below $0.5$.
		
		Let $t_2$ denote the first round $t_2>t_0$ such that $x_t\le 0.5$.
We first show that it will take no more than $k/2$ rounds for $x_t$ to drop below $0.5$:
  $t_2-t_0\le k/2$.
To this end, we use $t_1$ to denote the first round $t_1\ge t_0$ such that $y_t\ge 3/4$
  and note that $t_1\le t_2$ (since otherwise at $t=t_2-1$, we have $1/2\le y_t\le 3/4$
  and $1/2\le x_t\le e^6$ in order for $x_t$ to go below $1/2$ with $\eta\le 3$
  in the next round;
  this contradicts with the fact that the KL divergence is at least $20$ after $T_0$).

We break the proof of $t_2-t_0\le k/2$ into two phases:
  $t_1-t_0\le k/4$ and $t_2-t_1\le k/4$.

		{\bf Phase 1.} First we prove that it takes no more than $k/4$ steps for $y_t$ to get larger than $3/4$. 
To this end, we notice that for all $t\in [t_0:t_1-1]$, we have 
  $y_t\le 3/4$ and thus, $x_t\ge 3/4$ since the KL divergence is at least $20$. 
During all these rounds the loss vector $\ell_y^t$ of the second player satisfies
  $\smash{\ell^t_{y}(1) \leq -3/4 + 1/4 \leq -0.5}$ and $\smash{\ell^t_{y}(2) \geq 0.5}$. Thus we have (using $\smash{0.5\le y_{t_0}\le y_{t_1-1}\le 3/4}$)
		\begin{align*}
		3\ge \frac{y_{t_1-1}}{1 - y_{t_1-1}} \geq \exp\big(\eta(t_1-t_0-1)\big)\cdot \frac{y_{t_0}}{1 - y_{t_0}} \geq \exp\big(\eta(t_1-t_0-1)\big).
		\end{align*} 
		Thus $t_1-t_0\le (2/\eta)+1\le k/4$ using $k\ge 64/\eta$ and $\eta\le 3$.
		
		{\bf Phase 2.} Next we prove that, starting from $t_1$, it takes less than ${k}/4$ steps for $x_t$ to drop below $0.5$. Note that for each $t\in [t_1:t_2-1]$, 
		the loss vector $\ell^t_x$ of the first player satisfies
        $\smash{\ell^t_{x}(1) \geq  0.5}$ and $\smash{\ell^t_{x}(2)\leq-0.5}$. 
Moreover, we assume without loss of generality that 
  $\smash{1 - x_{t_1} \geq \exp(-(k+1)\eta/20)}$; otherwise
    the KL divergence at $t_1$ is already bigger than $10^{-10}(k + 1)\eta$ and we are   done.
    Therefore, 
		\begin{align*}
		1\le \frac{x_{t_2-1}}{1 - x_{t_2-1}} \leq \exp\big(-\eta(t_2-t_1-1)\big)\cdot \frac{x_{t_1}}{1 - x_{t_1}} \leq 
		\exp\big(\eta (-(t_2-t_1-1) + (k + 1)/20)\big)
		\end{align*}
		Thus $t_2-t_1\ge 1+(k+1)/20\le k/4$ using $k\ge 64/\eta\ge 64/3$.

Now we are at time $t_2$ and we examine the next 
  $R=3/\eta\le k/2$ rounds $[t_2:t_2+R]$;
  these are the rounds where we will gain a lot in the KL divergence. 
Given that $x_{t_2}$ just dropped below $1/2$,
  we have $x_{t_2}\ge 0.5\cdot \exp(-2\eta)$ and thus, for every $t\in [t_2:t_2+R]$, 
  $$
  x_t\ge x_{t_2}\cdot \exp(-2\eta \cdot R)\ge 0.5\cdot e^{-12}.
  $$
Consequently, we have \vspace{0.1cm}
		\begin{align*}
		&\hspace{-1cm}\big(\DKL(x^{\star}\| x^{t_2+R})+\DKL(y^\star\| y^{t_2+R})\big)
		- 
		\big(\DKL(x^{\star}\| x^{t_2})+\DKL(y^{\star}\| y^{t_2})\big) \\[0.5ex]
		\geq &~ \sum_{t= t_2}^{t_2+R-1}{e^{-7}}\eta^2 x_{t}(1-x_t)(2y_t-1)^2 +e^{-7}\eta^2 y_t(1-y_t)(2x_t-1)^2\\
		\geq &~\sum_{t = t_2}^{t_2+R - 1}e^{-7}\eta^2 x_{t}(1-x_t)(2y_t-1)^2 
		\geq~ \frac{3}{\eta}\cdot e^{-7}\eta^2\cdot \frac{1}{4}e^{-12}\cdot \frac{1}{4}
		\geq~ 10^{-10}\eta.
		\end{align*}
		So we conclude that after at most $k/4+k/4+k/2 = k$ steps, the KL divergence increase at least $10^{-10}\eta$. Thus at time $T_0+ k^2 + k \leq T_0+(k + 1)^2$, the KL divergence is at least $10^{-10}k\eta + 10^{-10}\eta$ $= 10^{-10} (k+1)\eta$. 
		This finishes the induction and the proof of the claim.
	\end{proof}

\section{Missing proof from Section~\ref{sec:correlated}}\label{appendixsec}
\subsection{Proof of Lemma \ref{lem:correlated-equilibrium1} }\label{sec:lem:correlated-equilibrium1}
 
Fix any swap function $\phi:[n]\rightarrow [n]$. 
By Lemma~\ref{lem:optimisticRegret}, every $\ALG_j$ achieves low regret. Thus, 
\begin{align}
\label{eq:correlated5}
\sum_{t\in [T]} \langle q^{t}_{j}, x^t(j)\ell^t\rangle \leq \sum_{t\in [T]}x^t(j)\cdot \ell^t(\phi(j)) + \frac{2\log n}{\eta} + \eta \sum_{t\in [T]} \|x^t(j)\ell^t - x^{t-1}(j)\ell^{t-1}\|_{\infty}^2,
\end{align}
where we used $x^t=Q^tx^t$, set $\ell^0=\mathbf{0}$ and $x^0=\mathbf{1}/n=x^1$. 
Consequently, we have
\begin{align}
    \sum_{t\in [T]}\langle x^{t},  \ell^{t}\rangle &= \sum_{t\in [T]}\langle x^{t} Q^{t}, \ell^{t}\rangle
    = \sum_{t\in [T]}\sum_{j\in [n]}\langle x^{t}(j)q^{t}_{j},\ell^t\rangle
    = \sum_{j\in [n]}\sum_{t\in [T]}\langle q^{t}_{j},x^{t}(j) \ell^t\rangle\notag\\[-0.5ex]
    &\leq \sum_{j\in [n]}\left(\sum_{t\in [T]} x^t(j) \cdot \ell^t(\phi(j)) + \frac{2\log n}{\eta} + \eta \sum_{t\in [T]} \| x^t(j)\ell^t - x^{t-1}(j)\ell^{t-1}\|_{\infty}^2 \right)\notag\\
    &=\sum_{t\in [T]}\sum_{j\in [n]}x^t(j)\cdot \ell^t(\phi(j)) + \frac{2n\log n}{\eta} + \eta \sum_{t\in [T]}\sum_{j\in [n]}\|x^t(j) \ell^t - x^{t-1}(j)\ell^{t-1}\|_{\infty}^2 \notag 
\end{align}
where the first inequality follows from \eqref{eq:correlated5}.
Furthermore, we have (using $\|\ell^t\|_\infty\le 1$ and $\|x^t\|_1=1$)
\begin{align}
    \sum_{j\in [n]}\|x^t(j)\ell^t - x^{t-1}(j)\ell^{t-1}\|_{\infty}^2 
    &\leq \sum_{j\in [n]}\Big(\|x^t(j)\ell^t - x^{t-1}(j)\ell^t\|_{\infty} + \|x^{t-1}(j)\ell^t - x^{t-1}(j)\ell^{t-1}\|_{\infty}\Big)^2\notag\\
    &\leq 2\sum_{j\in [n]}\|x^t(j)\ell^t - x^{t-1}(j)\ell^t\|_{\infty}^2 + 2\sum_{j\in [n]}\|x^{t-1}(j)\ell^t - x^{t-1}(j)\ell^{t-1}\|_{\infty}^2\notag\\
    &=2\sum_{j\in [n]}\left(x^t(j) - x^{t-1}(j)\right)^2 \|\ell^t\|_{\infty}^2 + 2\sum_{j\in [n]}(x^{t-1}(j))^2\|\ell^t - \ell^{t-1}\|_{\infty}^2\notag\\
    &= 2\Big(\|x^t - x^{t-1}\|_{2}^{2}\cdot \|\ell^t\|_{\infty}^2 + \|x^{t-1}\|_{2}^2\cdot  \|\ell^t - \ell^{t-1}\|_{\infty}^2\Big)\notag\\
    &\leq  2\Big(\|x^t - x^{t-1}\|_{1}^{2} + \|\ell^t - \ell^{t-1}\|_{\infty}^2\Big)\notag
\end{align}
We can combine all these inequalities (and note that $x^0=x^1$) to  finish the proof of the lemma.
 
\subsection{Proof of Lemma \ref{main-technical-lemma}}\label{lemma53}

We start the proof of Lemma \ref{main-technical-lemma}
  with the following definition. 
\begin{definition}
	\label{def:markov-multiplicative}
Given Markov chains $Q, Q' \in \R^{n\times n}$, we say $Q'$ is $(\eta_1, \ldots, \eta_n)$-approximate to $Q$ if 
	$(1 - \eta_i)q'_{i,j} \leq q_{i,j} \leq  (1 + \eta_i)q'_{i,j} $
	 for every $i, j \in [n]$,
	where we write $Q=(q_{i, j})$ and $Q'=(q_{i,j}')$. 
\end{definition}

We are ready to state our perturbation analysis on ergodic\footnote{Note that
  $Q^t$ used in \BMHedge\ is always ergodic.} Markov chains. 
\begin{lemma}
	\label{lem:perturbation}
	Given two ergodic Markov chains $Q$ and $Q'$, where $Q'$ is $(\eta_1, \ldots, \eta_n)$-approximate to $Q$, the stationary distribution $p, p'$ of $Q$ and $Q'$, respectively, satisfy $\|p - p'\|_1 \leq 8\sum_{i=1}^{n}\eta_i$.
\end{lemma}

The proof of Lemma~\ref{lem:perturbation} relies on the classical Markov chain tree theorem 
  (see \cite{anantharam1989proof}). To state it we need the following definition.

\begin{definition}
	\label{def:rooted-tree}
	Suppose $Q$ is an ergodic Markov chain and $G=(V, E)$ with
	$V=[n]$ is the weighted directed graph associated with $Q$.
We say a subgraph $T$ of $G$ is a  \emph{directed tree rooted at} $i\in [n]$ if
  (1) $T$ does not contain any cycles and (2) Node $i$ has no outgoing edges, while every other node $j\in [n]$ has exactly one outgoing edge.
	 For each node $i \in [n]$, we write $\mathcal{T}_{i}$ to denote the set of all directed trees rooted at node $i$. 
	We further define 
	\begin{align*}
	\Sigma_i = \sum_{T\in \mT_i}\prod_{(a , b ) \in T}q_{a,  b }\quad\text{and}\quad
	\Sigma=\sum_{i\in [n]} \Sigma_i,
	\end{align*}
	 i.e., the weight of $T$ is the product of its edge weights and
	 $\Sigma_i$ is the sum of weights of trees  in $\mT_i$.
\end{definition}
We can now formally state the Markov chain tree theorem.

\begin{theorem}[Markov chain tree theorem; see~\cite{anantharam1989proof}]
	\label{lem:Markov-chain-tree}
	Suppose $Q$ is an erogidc Markov chain and $p$ is its stationary distribution.
	Then  we have $p_i = \Sigma_i / \Sigma$ for every $i\in [n]$.
\end{theorem}

We now use the Markov chain tree theorem to prove Lemma~\ref{lem:perturbation}.
\begin{proof}[Proof of Lemma~\ref{lem:perturbation}]
Note that the lemma is trivial when $\sum_{i = 1}^{n}\eta_i >  {1}/{4}$ so 
we assume without loss of generality that $\sum_{i = 1}^{n}\eta_i \leq 1/4$.
For any $i \in [n]$, we have
\begin{align}
\Sigma_i &= \sum_{T\in \mT_i}\prod_{(a , b ) \in T}q_{a , b} \leq \sum_{T\in \mT_i}\prod_{(a, b) \in T}(1 + \eta_{a })\tilde{q}_{a, b } \notag\\[-1.5ex]
&\leq \prod_{j\in [n]}(1 + \eta_j)\sum_{T\in \mT_i}\prod_{(a , b ) \in T} {q}'_{a ,b} = \prod_{j\in [n]}(1 + \eta_j)\cdot  {\Sigma}'_i \leq \left(1 + 2\sum_{j \in [n]}\eta_j\right) {\Sigma}'_i. \label{eq:correlated3}
\end{align}
The third step holds because for any tree $T \in \mathcal{T}_i$, each node, other than node $i$,
  appears exactly once as $a$ when calculating the weight of $T$. 
The last step follows from the fact that when $\sum_{i=1}^{n}\eta_i\leq 1/4$,
\[
\prod_{j\in [n]} (1 + \eta_j) \leq \prod_{j\in [n]}e^{\eta_j} = e^{\sum_{j \in [n]}\eta_j} \leq 1 + 2\sum_{j \in [n]}\eta_j.
\]
Similarly, we have
\begin{align}
\Sigma_i &\geq \sum_{T\in \mT_i}\prod_{(a , b ) \in T}(1 - \eta_{a})\tilde{q}_{a, b }  
\ge \prod_{j\in [n]}(1 - \eta_j)\cdot  {\Sigma}'_i 
\geq \left(1 - 2\sum_{j \in [n]}\eta_j\right) {\Sigma}'_i. \label{eq:correlated4}
\end{align}
The last inequality holds since, for $\sum_{j=1}^{n}\eta_j \leq 1/2$, we have
\begin{align*}
\prod_{j\in [n]} (1 - \eta_j) \geq \prod_{j\in [n]}e^{-2\eta_j} = \exp\left(-2\sum_{j\in [n]}\eta_j\right) \geq 1 - 2\sum_{j \in [n]}\eta_j.
\end{align*}
Since $\Sigma = \sum_{i }\Sigma_{i}$, we have
$\left(1 - 2\sum_{i}\eta_i\right)\tilde{\Sigma} \leq \Sigma \leq \left(1 + 2\sum_{i}\eta_i\right)\tilde{\Sigma}
$.
Applying Theorem
 \ref{lem:Markov-chain-tree}, 
\begin{align*}
\|p -  {p}'\|_1 
&= \sum_{i\in [n]}|p_i - {p}'_i| 
=\sum_{i\in [n]} \Big| \Sigma_i \big/\Sigma -  {\Sigma_i}'\big/ { {\Sigma}'} \Big| 
\le  \sum_{i\in [n]} \Big| \Sigma_i \big/\Sigma - \Sigma_i\big/ { {\Sigma}'} \Big|  
  + \sum_{i\in [n]}\Big| \Sigma_i \big/ {\Sigma'} -  {\Sigma'_i}\big/ { {\Sigma}'} \Big| \\
&\leq \sum_{i\in [n]} \frac{2\sum_{i=1}^{n}\eta_i}{1 - 2\sum_{i=1}^{n}\eta_i}\Big| \Sigma_i /\Sigma \Big| + \sum_{i\in [n]}2 \sum_{j\in [n]} \eta_j\cdot \Big|  {\Sigma}'_i / {\Sigma'} \Big| \leq 6\sum_{i\in [n]}\eta_i.
\end{align*}
This finishes the proof of the lemma.
\end{proof}

Finally we prove Lemma \ref{main-technical-lemma}: 
\begin{proof}[Proof of Lemma \ref{main-technical-lemma}] 
We start with the following claim, which states that entries of $Q^{t}$ and $Q^{t-1}$ only differs 
  by a small  multiplicative factor. 
\begin{claim}
	\label{lem:optimistic-perturbation}
Suppose that the learning rate $\eta \leq 1/6$ and let $x^0=\mathbf{1}/n=x^1$.
Then for any $t\ge 2$,
  $Q^t$ is a $(\eta_1,\ldots,\eta_n)$-approximate to $Q^{t-1}$, where
$
\eta_j=2\eta x^{t-2}(j) + 4\eta x^{t-1}(j)
$ for each $j\in [n]$.
\end{claim}
Combing Claim \ref{lem:optimistic-perturbation} and Lemma~\ref{lem:perturbation}, we have
	\begin{align*}
	\|x^t - x^{t-1}\|_1 \leq 8\sum_{j\in [n]}\eta_j=8\sum_{j\in [n]}\left( 2x^{t-2}(j)+4x^{t-1}(j)\right)\eta = 48\eta.
	\end{align*}
This finishes the proof of Lemma \ref{main-technical-lemma}.
\end{proof}


\begin{proof}[Proof of Claim \ref{lem:optimistic-perturbation}]
Let $x^0=\mathbf{1}/n=x^1$. 
	By the updating rule of optimisitic Hedge,  we have for any $t\ge 2$, $i,j\in [n]$ that
	\begin{align*}
	q^{t}_{j}(i) 
	= &~ \frac{ \exp(-\eta (2x^{t-1}(j)\ell^{t-1}(i) - x^{t-2}(j)\ell^{t-2}(i))) \cdot q^{t-1}_{j}(i) } { \sum_{k\in [n]} \exp(-\eta (2x^{t-1}(j)\ell^{t-1}(k) - x^{t-2}(j)\ell^{t-2}(k))) \cdot q^{t-1}_{j}(k) }\\
	\leq &~ \frac{ \exp(\eta x^{t-2}(j) )\cdot  q^{t-1}_{j}(i) } { \sum_{k\in [n]} \exp(-2\eta x^{t-1}(j) ) \cdot q^{t-1}_{j}(k) }\\[0.7ex]
	= &~ \exp\big(\eta x^{t-2}(j) + 2\eta x^{t-1}(j) \big)\cdot  q^{t-1}_{j}(i)\\[0.6ex]
	\leq &~ (1 + 2\eta x^{t-2}(j) + 4\eta x^{t-1}(j) )\cdot q^{t - 1}_{j}(i).
	\end{align*}
	The second step follows from $\ell^{t} \in [0, 1]^{n}$ and the last step follows from $\exp(a) \leq 1 + 2a$ for $a \leq 1/{2}$. The other side holds similarly:
	\begin{align*}
	q^{t}_{j}(i) 
	= &~ \frac{ \exp(-\eta (2x^{t-1}(j)\ell^{t-1}(i) - x^{t-2}(j)\ell^{t-2}(i)))\cdot q^{t-1}_{j}(i) } { \sum_{k\in [n]} \exp(-\eta (2x^{t-1}(j)\ell^{t-1}(k) - x^{t-2}(j)\ell^{t-2}(k)))\cdot q^{t-1}_{j}(k) }\\
	\geq &~ \frac{ \exp(-2\eta x^{t-1}(j) )\cdot  q^{t-1}_{j}(i) } { \sum_{k\in [n]} \exp(\eta x^{t-2}(j) ) \cdot q^{t-1}_{j}(k) }\\[0.7ex]
	= &~ \exp\big(- \eta x^{t-2}(j) - 2\eta x^{t-1}(j) \big)\cdot  q^{t-1}_{j}(i)\\[0.6ex]
	\geq &~ (1 - \eta x^{t-2}(j) - 2\eta x^{t-1}(j) )\cdot  q^{t - 1}_{j}(i).
	\end{align*}
	Thus completing the proof.
\end{proof}	
\subsection{Proof of Corollary \ref{cor:robus-adv}}\label{corollary54}

 The algorithm works as follow. We set $$\eta = \frac{(n\log n)^{1/4}} {m^{1/2}T^{1/4}}$$ and $B_r = 1$ at initialization, for any player $i \in [m]$ and $\tau = 1, \ldots, T$
\begin{enumerate}
	\item Play $x^{t}_{i}$ according to \BMHedge,  and receive $\ell_{i}^{t}$.
	\item If $\sum_{t=2}^{\tau} \|\ell_{i}^{t} - \ell_{i}^{t-1}\|_{\infty}^{2} + \sum_{t=2}^{\tau}\|x_{i}^{t} - x_{i}^{t-1}\|_{1}^{2}  \geq B_r$.
	\begin{enumerate}
		\item Update $B_{r+1} = 2B_{r}$, $r \leftarrow r + 1$, $\eta_{r} = \min\left\{\sqrt{\frac{n\log n}{B_r}}, \eta\right\}$.
		\item Start a new run of \BMHedge\ with learning rate $\eta_{r}$.
	\end{enumerate}
\end{enumerate}
 
For any round $r$, we use $T_r$ to denote its final iteration and \[
I_r = \sum_{t=T_{r-1} + 1}^{T_r}\|x^{t}_i - x^{t-1}_{i}\|_1^2 + \sum_{t=T_{r-1} + 1}^{T_r}\|\ell^{t}_i - \ell^{t-1}_{i}\|_{\infty}^2.
\] 
Then we have
\begin{align*}
\sregret_{T_{r-1} + 1: T_r} \leq &~ \frac{2n\log n}{\eta_r} + 2\eta_r \left(\sum_{t=T_{r-1} + 1}^{T_r}\|x^{t}_i - x^{t-1}_{i}\|_1^2 + \sum_{t=T_{r-1} + 1}^{T_r}\|\ell^{t}_i - \ell^{t-1}_{i}\|_{\infty}^2  \right)\\
\leq & ~ 2(n\log n)^{3/4} \cdot T^{1/4}m^{1/2} + 2\sqrt{n\log nB_r} + 2\eta_r \cdot I_r \\
\leq & ~ 2(n\log n)^{3/4} \cdot T^{1/4}m^{1/2}  + 2\sqrt{n\log nB_r}   + 2\sqrt{2n\log n I_r}\\
\leq & ~ 2(n\log n)^{3/4} \cdot T^{1/4}m^{1/2}  + 4\sqrt{2n\log n I_r}\\
\leq & ~ 2(n\log n)^{3/4} \cdot T^{1/4}m^{1/2}  + 4\sqrt{2n\log n} \cdot\sqrt{ \left(\sum_{t=2}^{T}\|x^{t}_i - x^{t-1}_{i}\|_1^2 + \sum_{t=2}^{T}\|\ell^{t}_i - \ell^{t-1}_{i}\|_{\infty}^2  \right)}
\end{align*}
The first step follows from Lemma~\ref{lem:correlated-equilibrium1}, the second step follows from the definition of $I_r$ and the fact 
\[
\frac{1}{\eta_r} \leq \frac{1}{\eta} + \sqrt{\frac{B_r}{n\log n}} =  \frac{m^{1/2}T^{1/4}}{(n\log n)^{1/4}} + \sqrt{\frac{B_r}{n\log n}} 
\]
The third step follows from 
$
\eta_r \leq \sqrt{\frac{n\log n}{B_r}} \leq \sqrt{\frac{n\log n}{I_r/2}},
$
and the last step comes from $\sqrt{B_r} \leq \sqrt{2 I_r}$.

Since the number of round is at most $O(\log T)$, we have
\[
\sregret_{T} \leq \log T \left(2(n\log n)^{3/4} T^{1/4}m^{1/2}  + 4\sqrt{2n\log n} \cdot \sqrt{2 \left(\sum_{t=1}^{T}\|x^{t}_i - x^{t-1}_{i}\|_1^2 + \sum_{t=1}^{T}\|\ell^{t}_i - \ell^{t-1}_{i}\|_{\infty}^2  \right)}\right)
\]

If all players adopt the algorithm, then we know their learning rate is no greater than $\eta = \frac{(n\log n)^{1/4}} {m^{1/2}T^{1/4}}$, thus we know $\|x^{t}_i - x^{t-1}_i\|_1 \leq O(\eta) =   O\left(\frac{(n\log n)^{1/4}} {m^{1/2}T^{1/4}} \right)$ (see Lemma~\ref{main-technical-lemma}) and $\|\ell^{t}_i - \ell^{t-1}_i\|_{\infty} \leq \sum_{j \neq i}\|x^{t}_{j} - x^{t-1}_j\|_1 \leq m\cdot O(\eta) = O\left(\frac{m^{1/2}(n\log n)^{1/4}} {T^{1/4}}\right)$. Thus the swap regret is at most 
\[
O\left((n\log n)^{3/4}m^{1/2}T^{1/4}\log T\right).
\]

If the player is facing an adversary, then $\|x^{t}_i - x^{t-1}_i\|_1 \leq 2$ and $\|\ell^{t}_i - \ell^{t-1}_i\|_{\infty} \leq 1$, thus we conclude its regret is at most \[
O\left(\sqrt{n\log nT}\log T + (n\log n)^{3/4}m^{1/2}T^{1/4}\log T\right ).
\]

\section{Another no swap regret algorithm}
\label{sec:extension-phi}

We prove the optimistic variant of a folklore algorithm, originally appeared in~\cite{cesa2006prediction}, could also achieve fast convergence of swap regret. Our perturbation analysis again plays a key role in the regret analysis.

Define $\Phi$ to be all swap functions that map $[n]$ to $[n]$. We have $|\Phi| = n^{n}$.  
For any $\phi \in \Phi$, define the swap matrice $S^{\phi}$ as: $S^{\phi}_{i, j} = 1$ if $\phi(i) = j$ and $S^{\phi}_{i, j} = 0$ otherwise. It is easy to see that $S^{\phi}$ contains exactly one $1$ each row. 

~\cite{cesa2006prediction} treats each swap matrice $S^{\phi}$ as an expert, and run Hedge algorithm on all $n^{n}$ swap matrices. 
At time $t$, the output strategy $p^{t}$ is determined by these experts via solving a fix point problem\footnote{The algorithm is not efficient in general. However, we can turn it into an effiecient one by considering only $n^2$ swap matrices that are equal to indentical mapping {\em except} for one coordinate. The regret bound will only blow up by a $\sqrt{n}$ factor.}.
The optimisitic variant of~\cite{cesa2006prediction} is shown in Algorithm~\ref{algo:optimistic-meta-exp}. We first analysis the regret, 

\begin{algorithm}[!h]
	\caption{}
	\label{algo:optimistic-meta-exp}
	\begin{algorithmic}[1]
		\For{$t = 1, 2, \ldots,$}
		\State Play $p^t$ and receive the loss vector $l^t$.
		\State Update 
		\begin{align*}
		q^{t+1}(\phi) = \frac{x^{t}(\phi)\exp(-\eta (2x^{t} S^{\phi} \ell^{t}  - x^{t-1} S^{\phi} \ell^{t-1}) )}{\sum_{\phi\in \Phi}x^{t}(\phi)\exp(-\eta (2x^{t} S^{\phi} \ell^{t} - x^{t-1} S^{\phi} \ell^{t-1}) )} \quad \forall \phi\in \Phi
		\end{align*}
		\State Compute $x^{t+1} = x^{t+1}Q^{(t+1)}$, where 
		\begin{align*}
		Q^{(t + 1)} = \sum_{\phi \in \Phi}q^{t + 1}(\phi)S^{\phi}.
		\end{align*}
		\EndFor
	\end{algorithmic}
\end{algorithm}

\begin{lemma}
	\label{lem:meta-optimistic-regret}
	Algorithm~\ref{algo:optimistic-meta-exp} achieves regret
	\begin{align*}
	\sregret_{T} \leq \frac{n\log n}{\eta} + 2\eta \sum_{t=2}^{T}\|x^t - x^{t-1}\|_1^2 + 2\eta \sum_{t=2}^{T}\|\ell^t - \ell^{t-1}\|_{\infty}^{2}.
	\end{align*}
\end{lemma}
\begin{proof}
	According to the updating rule, for any $\phi \in \Phi$, we have
	\begin{align*}
	\sregret_{T} =&~ \sum_{t=2}^{T}\langle x^{t}, \ell^{t}\rangle - \max_{\phi\in \Phi}\sum_{t=2}^{T}x^{t}S^{\phi}\ell^t \\
	=&~ \sum_{t=2}^{T}\langle x^{t}Q^{(t)}, \ell^{t}\rangle - \max_{\phi\in \Phi}\sum_{t=2}^{T}x^{t}S^{\phi}\ell^t\\
	=&~ \sum_{t=2}^{T}\sum_{\phi \in \Phi}x^t(q^{t}(\phi) S^{\phi})\ell^t - \max_{\phi\in \Phi}\sum_{t=2}^{T}x^{t}S^{\phi}\ell^t \\
	=& \sum_{t=2}^{T}\sum_{\phi \in \Phi}q^{t}(\phi)\cdot  x^t S^{\phi} \ell^t - \max_{\phi\in \Phi}\sum_{t=2}^{T}x^{t}S^{\phi}\ell^t \\
	\leq &~ \frac{n\log n}{\eta} + \eta\sum_{t=2}^{T}\max_{\phi\in\Phi}\left|x^{t}S^{\phi}\ell^{t-1} - x^{t-1}S^{\phi}\ell^{t-1} \right\|^2\\
	\leq &~\frac{\log n}{\eta} + 2\eta \sum_{t=2}^{T}\|x^t - x^{t-1}\|_1^2 + 2\eta \sum_{t=2}^{T}\|\ell^t - \ell^{t-1}\|_{\infty}^{2}.
	\end{align*}
	The fifth step follows the regret bound of optimistic Hedge and the last step follows from the fact that for any $\phi \in \Phi$,
	\begin{align*}
	\left|x^{t}S^{\phi}\ell^t - x^{t}S^{\phi}\ell^t \right|^2 
	= &~ \left|x^t S^{\phi}\ell^t - x^{t-1}S^{\phi}\ell^t + x^{t-1}A_{\phi}\ell^t - x^{t-1}S^{\phi}\ell^{t-1} \right|^2\\
	\leq &~2\left|x^{t}S^{\phi}\ell^t - x^{t-1}S^{\phi}\ell^t|^2 + 2|x^{t-1}S^{\phi}\ell^t - x^{t-1}S^{\phi}\ell^{t-1} \right|^2\\
	=&~2\langle x^t - x^{t-1} , S^{\phi}\ell^{t}\rangle + 2\langle x^{t-1}S^{\phi}, \ell^{t}-l^{t-1} \rangle\\
	\leq &~2\|x^t - x^{t-1} \|_1^2 \|S^{\phi}\ell^{t}\|_{\infty}^{2} + 2\|x^{t-1}S^{\phi}\|_{1}\|\ell^{t} - \ell^{t-1}\|_{\infty}^2\\
	\leq & 2\|x^t - x^{t-1}\|_1^2 + 2\|\ell^t - \ell^{t-1}\|_{\infty}^{2}. 
	\end{align*}
	Thus completing the proof.
\end{proof}
It remains to show that the environment is stable. Again, since $x^{t}$ is the stationary distribution of $Q^{(t)}$, we only need some perturbation analysis on $Q^{(t)}$. In particular, we have
\begin{lemma}
	\label{lem:perturb-meta}
	For any $t$, $Q^{(t)}$ is $(6\eta, \ldots, 6\eta)$ approximate to $Q^{(t + 1)}$.
\end{lemma}
\begin{proof}
	For any $\phi$, we have
	\begin{align*}
		q^{t+1}(\phi) 
		=&~ \frac{q^{t}(\phi)\exp(-\eta (2x^{t} A_{\phi} \ell^{t}  - x^{t-1} A_{\phi} \ell^{t-1}) )}{\sum_{\phi\in \Phi}q^{t}(\phi)\exp(-\eta (2x^{t} A_{\phi} \ell^{t} - x^{t-1} A_{\phi} \ell^{t-1}) )}\\
		\leq &~ \frac{q^{t}(\phi) \exp(\eta)}{\sum_{\phi\in \Phi} q^{t}(\phi) \exp(-2\eta)}\\
		\leq &~ (1 + 6 \eta)q^{t}(\phi)
	\end{align*}
	Similarly, we have
	\begin{align*}
		q^{t+1}(\phi) 
		=&~ \frac{q^{t}(\phi)\exp(-\eta (2x^{t} A_{\phi} \ell^{t}  - x^{t-1} A_{\phi} \ell^{t-1}) )}{\sum_{\phi\in \Phi}q^{t}(\phi)\exp(-\eta (2x^{t} A_{\phi} \ell^{t} - x^{t-1} A_{\phi} \ell^{t-1}) )}\\
		\geq &~ \frac{q^{t}(\phi) \exp(-2\eta)}{\sum_{\phi\in \Phi} q^{t}(\phi) \exp(\eta)}\\
		\geq &~ (1 - 6 \eta)q^{t}(\phi)
	\end{align*}
	Thus, for any $i, j\in [n]$, we have
	\begin{align*}
	Q_{i, j}^{(t+1)} = \sum_{\phi \in \Phi}q^{t+1}(\phi)S^{\phi}_{i, j} \leq (1 + 6\eta ) \sum_{\phi \in \Phi}q^{t}(\phi)S^{\phi}_{i, j} = (1 + 6\eta) Q_{i, j}^{(t)}
	\end{align*}
	and 
	\begin{align*}
		Q_{i, j}^{(t+1)} = \sum_{\phi \in \Phi}q^{t+1}(\phi)S^{\phi}_{i, j} \geq (1 - 6\eta ) \sum_{\phi \in \Phi}q^{t}(\phi)S^{\phi}_{i, j} \geq  (1 - 6\eta) Q_{i, j}^{(t)}
	\end{align*}
	Thus we conclude $Q^{(t)}$ is $(6\eta, \ldots, 6\eta)$ approximate to $Q^{(t + 1)}$.
\end{proof}
Combining the above results, we have
\begin{theorem}
	\label{thm:correlated-equilibrium-ext}
	Suppose every player uses Algorithm~\ref{algo:optimistic-meta-exp} and choose $\eta = O\left((\frac{\log n}{nm^2T})^{1/4}\right)$, then each individual's swap regret is at most $O\left( m^{1/2}n^{5/4}(\log n)^{3/4}T^{1/4} \right)$.
\end{theorem}

\begin{proof}
	By Lemma~\ref{lem:meta-optimistic-regret}, for any palyer $i \in [m]$, we have
	\begin{align*}
	\sregret_{T} \leq &~ \frac{n\log n}{\eta} + 2\eta \sum_{t=2}^{T}\|x_i^t - x_i^{t-1}\|_1^2 + 2\eta \sum_{t=2}^{T}\|\ell_i^t - \ell_i^{t-1}\|_{\infty}^{2} \\
	\leq &  \frac{n\log n}{\eta} + 2\eta \sum_{t=2}^{T}\|x^t - x^{t-1}\|_1^2 + 2m\eta \sum_{t=2}^{T}\sum_{j\neq i}\|x^t_{j} - x^{t-1}_{j}\|_{1}^{2}
	\end{align*}
	where $w^{t}$ denotes the other player's strategy. Moreover, since $Q^{(t-1)}$ is $(6\eta, \ldots, 6\eta)$ approximates to $Q^{(t)}$, we know
	\begin{align*}
	\|x^{t}_{i} - x^{t-1}_{i}\|_1 \leq 8\cdot \sum_{i=1}^{n}6\eta = O(n\eta)
	\end{align*}
	holds for any $i$. Thus we have
	\begin{align*}
	\sregret_{T} \leq &~ \frac{n\log n}{\eta} +  2\eta \sum_{t=2}^{T}\|x^t - x^{t-1}\|_1^2 + 2m\eta \sum_{t=2}^{T}\sum_{j\neq i}\|x^t_{j} - x^{t-1}_{j}\|_{1}^{2}\\
	\leq &~ \frac{n\log n}{\eta} + O(\eta^3n^2m^2T).
	\end{align*}
	Choosing $\eta = O\left((\frac{\log n}{nm^2T})^{1/4}\right)$, the regret is 
	\[
	\sregret_{T} = O\left( n^{5/4}(\log n)^{3/4}T^{1/4}m^{1/2} \right).
	\]
\end{proof}

\section{Price of anarchy}
\label{sec:price}

In this section, we show that a large class of no swap regret algorithm satisfies the {\em low approximate regret} property (see Definition~\ref{def:low-approximate}). 
Thus when all players adopt such algorithm, they experience fast convergence to an approximately optimal social welfare in {\em smooth games} (see Definition~\ref{def:smooth-game}).
In particular, we show that the average social welfare converges to an approximately optimal welfare at rate $O(1/T)$. 
The proof in this section is straightforward, our aim is to point out that such fast convergence rate generally holds for no-swap regret algorithms.
We first introduce the smooth game. Recall $\mathcal{L}(\bx) = \sum_{i\in [m]}\mathcal{L}_i(\bx)$ is the summation of each individual's loss under strategy profile $\bx$.
\begin{definition}[Smooth game]
	\label{def:smooth-game}
	A cost minimization game is $(\lambda, \mu)$-smooth if for all strategy profiles $\bx$ and $\bx^{\star}$, $\sum_{i}\mathcal{L}_i (x_i^{\star}, x_{-i}) \leq \lambda\cdot \mathcal{L}(\bx^{\star}) + \mu \cdot\mathcal{L}(\bx)$.
\end{definition}
A wide range of games belongs to smooth game, including routing games, auctions, etc. We refer interested reader to \cite{roughgarden2015intrinsic} for detailed coverage.

We next introduce the definition of low approximate regret.
\begin{definition}[Low approximate regret~\cite{foster2016learning}]
	\label{def:low-approximate}
	A learning algorithm satisfies the low approximate regret property for given parameters $(\eps, A(n))$ , if 
	\begin{align*}
	(1 - \eps)\sum_{t=1}^{T}\langle x^t, \ell^t \rangle \leq \min_{i} L(i) + \frac{A(n)}{\eps}.
	\end{align*}
\end{definition}

\begin{lemma}
	\label{lem:low-approximate}
	The BM reduction transfers the low approximate regret property. In particular, if we reduce from a no external regret algorithm satisfying low approximate regret with$(\eps, A(n))$, then the no swap regret algorithm satisfies low approximate regret with $(\eps, nA(n))$.
\end{lemma}
\begin{proof}
		 For any fixed $i$, using the low approximate regret property, we know
		\begin{align*}
		(1 - \eps)\sum_{t=1}^{T}\langle q^{t}_{j}, x^t(j)\ell_t\rangle \leq \min_{i'}\sum_{t=1}^{T}x^t(j)\ell^t(i') + \frac{A(n)}{\eps} \leq \sum_{t=1}^{T}x^t(j)\ell_t(i) + \frac{A(n)}{\eps}.
		\end{align*}
		Consequently, we have
		\begin{align*}
		(1 - \eps)\sum_{t=1}^{T}\langle x^{t},  \ell^{t}\rangle &= (1 - \eps)\sum_{t=1}^{T}\langle x^{t} Q^{(t)}, \ell^{t}\rangle\\
		&= (1 - \eps)\sum_{t=1}^{T}\sum_{j=1}^{n}\langle x^{t}(j)q^{t}_{j},\ell^t\rangle\\
		&= (1 - \eps)\sum_{j=1}^{n}\sum_{t=1}^{T}\langle q^{t}_{j},x^{t}(j) \ell^t\rangle\\
		&\leq \sum_{j=1}^{n}\left(\sum_{t=1}^{T} x^t(j) \ell^t(i) + \frac{A(n)}{\eps}\right)\\
		&=\sum_{t=1}^{T}\sum_{j=1}^{n}x^t(j)\ell^t(i) + \frac{nA(n)}{\eps}\\
		&=\sum_{t=1}^{T}\ell^t(i) + \frac{nA(n)}{\eps}.
		\end{align*}
		Thus concluding the proof.
\end{proof}

A direct corollary of Lemma~\ref{lem:low-approximate} and Theorem 3 in~\cite{foster2016learning} is  
\begin{theorem}
	\label{thm:price-of-anarchy}
	In a $(\lambda, \mu)$-smooth game, if all players use no swap regret algorithm generated from BM reduction and a no external regret algorithm  satisfying low approximate regret property with parameter $\eps$ and $A(n) = \log n$, then we have
	\begin{align*}
	\frac{1}{T}\sum_{t=1}^{T}\mathcal{L}(\bx_t) \leq \frac{\lambda}{1 - \mu - \eps}\cdot \OPT + \frac{m}{T}\cdot \frac{1}{1 - \mu - \eps} \cdot \frac{n\log n}{\eps}.
	\end{align*}
	where $\OPT$ denotes the optimal social welfare, i.e., $\min_{\bx}\mathcal{L}(\bx)$.
\end{theorem}

\end{document}